\documentclass[journal,twoside]{IEEEtran}
\ifCLASSINFOpdf
\else
\fi

\usepackage{amsmath}   

\usepackage{amssymb}   
\usepackage{amscd} 
\usepackage{latexsym}  
\usepackage{epsfig}
\usepackage{bbm}
\usepackage{color}

\DeclareMathOperator{\conv}{conv}

\newcommand{\B}{{\cal B}}

\newcommand{\ee}{\end{equation}}
\newcommand{\be}{\begin{equation}}

\def\epsilon{{\varepsilon}}

\def\@begintheorem#1#2{\tmpitemindent\itemindent\topsep 0pt\rm\trivlist
    \item[\hskip \labelsep{\indent\it #1\ #2:}]\itemindent\tmpitemindent}
\def\@opargbegintheorem#1#2#3{\tmpitemindent\itemindent\topsep 0pt\rm \trivlist
    \item[\hskip\labelsep{\indent\it #1\ #2\ \rm(#3)}]\itemindent\tmpitemindent}
\def\@endtheorem{\endtrivlist\unskip}

\def\BibTeX{{\rm B\kern-.05em{\sc i\kern-.025em b}\kern-.08em
    T\kern-.1667em\lower.7ex\hbox{E}\kern-.125emX}}
\def\={\triangleq}                           

\newcommand{\qed}{\hfill $\blacksquare$}
\newcommand{\tr}{{\operatorname{Tr}\,}}

\newcommand{\bra}[1]{{\langle{#1}|}}
\newcommand{\ket}[1]{{|{#1}\rangle}}

\newcommand{\C}{{\mathcal{C}}}

\newcommand{\N}{{\mathbb{N}}}
\newcommand{\X}{{{\mathcal X}}}

\newcommand{\fset}[1]{{\mathcal{#1}}}
\newcommand{\V}{\mathcal V}

\newcommand{\W}{\mathcal W}
\newcommand{\M}{\mathcal M}
\newcommand{\A}{\mathcal A}
\newcommand{\D}{\mathcal D}

\renewcommand{\P}{\mathcal P}

\newcommand{\1}{{\mathbbm{1}}}

\newlength{\blank}
\newenvironment{proof}[1][{\hspace{-\blank}}]{{\noindent\textbf{Proof~{#1}.\ }}}{\hfill\qed}

\newtheorem{Theorem}{Theorem}[section]
\newtheorem{Definition}[Theorem]{Definition}
\newtheorem{rem}[Theorem]{Remark}
\newtheorem{Corollary}[Theorem]{Corollary}

\newtheorem{Lemma}[Theorem]{Lemma}
\newtheorem{Proposition}[Theorem]{Proposition}

\newcommand{\nc}{\newcommand}
\nc{\rnc}{\renewcommand}
\nc{\beq}{\begin{equation}}
\nc{\eeq}{{\end{equation}}}
\nc{\beqa}{\begin{eqnarray}}
\nc{\eeqa}{\end{eqnarray}}
\nc{\lbar}[1]{\overline{#1}}
\nc{\proj}[1]{| #1\rangle\!\langle #1 |}
\nc{\avg}[1]{\langle#1\rangle}
\nc{\Rank}{\operatorname{Rank}}
\nc{\smfrac}[2]{\mbox{$\frac{#1}{#2}$}}
\nc{\ox}{\otimes}
\nc{\dg}{\dagger}
\nc{\dn}{\downarrow}
\nc{\cA}{\mathcal{A}}
\nc{\cB}{\mathcal{B}}
\nc{\cC}{\mathcal{C}}
\nc{\cD}{\mathcal{D}}
\nc{\cE}{\mathcal{E}}
\nc{\cF}{\mathcal{F}}
\nc{\cG}{\mathcal{G}}
\nc{\cH}{\mathcal{H}}
\nc{\cI}{\mathcal{I}}
\nc{\cJ}{\mathcal{J}}
\nc{\cK}{\mathcal{K}}
\nc{\cL}{\mathcal{L}}
\nc{\cM}{\mathcal{M}}
\nc{\cN}{\mathcal{N}}
\nc{\cO}{\mathcal{O}}
\nc{\cP}{\mathcal{P}}
\nc{\cR}{\mathcal{R}}
\nc{\cS}{\mathcal{S}}
\nc{\cT}{\mathcal{T}}
\nc{\cU}{\mathcal{U}}
\nc{\cX}{\mathcal{X}}
\nc{\cY}{\mathcal{Y}}
\nc{\cZ}{\mathcal{Z}}
\nc{\csupp}{{\operatorname{csupp}}}
\nc{\qsupp}{{\operatorname{qsupp}}}
\nc{\var}{\operatorname{var}}
\nc{\rar}{\rightarrow}
\nc{\lrar}{\longrightarrow}
\nc{\polylog}{\operatorname{polylog}}

\nc{\RR}{{{\mathbb R}}}
\nc{\CC}{{{\mathbb C}}}
\nc{\FF}{{{\mathbb F}}}
\nc{\NN}{{{\mathbb N}}}
\nc{\ZZ}{{{\mathbb Z}}}
\nc{\PP}{{{\mathbb P}}}
\nc{\QQ}{{{\mathbb Q}}}
\nc{\UU}{{{\mathbb U}}}
\nc{\EE}{{{\mathbb E}}}

\begin{document}

\title{Secure and Robust Identification via Classical-Quantum Channels}

\makeatletter 

\author{Holger Boche~\IEEEmembership{Fellow,~IEEE}\thanks{HB is with Lehrstuhl f\"ur Theoretische
Informationstechnik, Technische Universit\"at M\"unchen, M\"unchen and Munich Center for Quantum Science and Technology (MCQST), M\"unchen, Germany. boche@tum.de}, %
Christian Deppe~\IEEEmembership{Member,~IEEE}\thanks{CD is with Lehrstuhl f\"ur Nachrichtentechnik,
Technische Universit\"at M\"unchen, M\"unchen, Germany. christian.deppe@tum.de}, %
and Andreas Winter\thanks{AW is with ICREA and with Departament de F\'{\i}sica: Grup d'Informaci\'{o} Qu\`{a}ntica, Universitat Aut\`{o}noma de Barcelona, ES-08193 Bellaterra (Barcelona), Spain. andreas.winter@uab.cat}}


\maketitle
\makeatother 

\begin{abstract}
We study the identification capacity of classical-quantum channels
(``cq-channels'') under channel uncertainty and privacy constraints. 
To be precise, we first consider compound memoryless cq-channels 
and determine their identification capacity; then we add an eavesdropper by 
considering compound memoryless wiretap cqq-channels, and determine 
their secret identification capacity. In the first case (without privacy), 
we find the identification capacity always equal to the transmission 
capacity. In the second case, we find a dichotomy: either the secrecy capacity
(also known as private capacity) of the channel is zero, 
and then the secrecy identification capacity is also zero, 
or the secrecy capacity is positive and then the secrecy 
identification capacity equals the transmission capacity of the main
channel without the wiretapper.
We perform the same analysis for the case of arbitrarily varying 
wiretap cqq-channels (cqq-AVWC) with analogous findings, and make several 
observations regarding the continuity and super-additivity of the 
identification capacity in the latter case.
\end{abstract}

\section{Introduction}
\label{Introduction}
Identification via channels was introduced by Ahlswede and Dueck \cite{AhlswedeDueck:ID.A}
forty years after Shannon \cite{Shannon}
introduced information theory as a theory of communication.  
In Shannon's transmission theory, the sender encodes the messages
as sequences of channel input letters in such a way, that although the channel
might not transmit the sequence correctly, the receiver still can decide 
what message had been sent, at least with a high probability.

In the theory of identification, the
receiver is not interested in the exact message, but only wants to know if 
the sent message is equal to a particular one that he is interested in. 
Of course, the sender does not know in which message the receiver is interesting.
It was shown that there are codes for classical channels 
with double exponential size in the block length of the codewords. 
In identification theory, one considers also models in which several receivers 
receive the same transmission but are interested in different one messages.
Applications for identification codes can be found in the theory of digital 
watermarks \cite{AhlswedeCai,P01} and communication complexity \cite{T01}.

Investigation into communication via quantum channels started in the 1960s. 
We refer the reader to the book~\cite{Wilde:book} for more details on quantum 
and classical channels and the various transmission capacities associated 
with them, including their history.

L\"ober \cite{Loeber:PhD} was the first to consider identification 
via classical-quantum channels (so-called cq-channels). 
He introduced two generalizations of the classical identification codes. 
First he defined identification codes for cq-channels where the receiver
has a binary measurement for each possible message he could be interested
in. Crucially, in quantum mechanics, these measurements may be
incompatible, meaning that one cannot identify several messages at the same time.
In certain applications, this is an undesirable feature, when there
are many receivers each wanting to identify ``their'' message. To
address this, L\"ober formulated a second model, that of a
\emph{simultaneous ID-code}, for which there has to be one single 
(simultaneous) measurement that allows us to identify every message at the same 
time. This model is also valid if the one who performs the measurement is not 
the ultimate receiver, and in particular, does not know in which
message this receiver is interested. 
There are many examples where identification schemes require 
simultaneous ID-codes because their real implementation would consist of many 
receivers at a time (for examples see \cite{AhlswedeDueck:ID.A}). 
This is not always the case \cite{Kleinew} if both
sender and receiver have a (possibly different) text and they want to check if
it is the same one, using an ID-code. Here, only one receiver is asking
only one question.

In the present paper, we consider both \emph{secure} and \emph{robust} models of 
cq-channels. Our coding schemes are all simultaneous,
but we will prove converses in the general, non-simultaneous setting.
With this, we characterize the identification and the simultaneous identification capacity. Here, these two capacities turn out to be the same.

Security is modeled by a channel with an eavesdropper, called a wiretap cqq-channel. 
It is connecting a sender with two receivers, a legal one and a wiretapper. 
The legitimate receiver accesses the output of the first channel 
and the wiretapper observes the output of the second channel.
A code for the channel conveys information to the
legal receiver such that the wiretapper knows nothing about
the transmitted  information.  The classical degraded form of this channel was introduced
by Wyner \cite{W75}, who determined the secrecy capacity of this channel. 
The classical non-degraded model was presented and solved in \cite{CK78}. 
The wiretap cqq-channel was considered in \cite{CWY04} and in \cite{D05}. 

To model the robustness aspect, we consider compound cq-channels, 
which are described as a set of memoryless channels. 
Before the start of the transmission, a channel is chosen unknown to
the sender or receiver, and used during the transmission of one codeword.
The code of the sender and the 
receiver has to be robust and therefore independent of the chosen model. 
The classical channel model was introduced by Blackwell, Breiman, and 
Thomasian \cite{BBT59}.
The compound cq-channel was considered in \cite{BB09}, \cite{M15} and \cite{BJK17}.

There exist many combinations of these concepts.
The classical compound wiretap channel was considered in \cite{LKPS09} and \cite{BBS11}.
The transmission capacity of the compound wiretap cqq-channel was given 
in \cite{BCCD14}. For an overview we refer to the wide ranging textbook
by Wilde \cite{Wilde:book}, which only omits the theory of identification over 
quantum channels. An overview on this topic can be found in \cite{winter:survey}. 
In \cite{BD17} we gave the identification capacities for the classical compound 
channel and the classical compound wiretap channel. Therefore, the present paper 
is a generalization to the classical-quantum case. 

The structure of our paper is as follows.
We start in Section~\ref{Basic} with the basic definitions of cq-channels and 
of transmission and identification via cq-channels; we
review the main result of \cite{Loeber:PhD} and \cite{AhlswedeWinter:ID-q}
where the identification capacity of cq-channels were given. 
We generalize this result in Section~\ref{Robust} for identification via compound cq-channels.
In Section~\ref{Secure}, we define how to add the wiretapper to the model, define wiretap cqq-channels 
and give their capacity, and we prove a dichotomy theorem for its secure 
identification capacity. 
We generalize this result in Section~\ref{Securerobust} for the secure identification 
capacity of a compound wiretap cqq-channel, i.e.~we prove a capacity theorem
for secure and robust identification via quantum channels.
In Section~\ref{Jammer} we assume that the channel state can change after 
each qubit transmitted by the sender. 
We assume that this action comes from a jammer and consider the worst case. 
In Section~\ref{Securejammer} we
also add a wiretapper to this model. We give the capacity for both models. 
Finally, in Section~\ref{Continuity} we analyze the calculated capacities as functions 
of the channel parameter.

\section{Basic definitions and results}
\label{Basic}
In this section we give recall the definitions of cq-channels, and of 
transmission and identification via cq-channels. 
Furthermore, we review the main results of \cite{Loeber:PhD} 
and \cite{AhlswedeWinter:ID-q}. 

Cq-channels have a classical sender, having access to an input 
alphabet $\X$, but their output is quantum, being described
by a Hilbert space $\B$. 
As is customary, we identify the states on $\B$, ${\cS}(\B)$
with the set of density operators,  i.e.~the self-adjoint, positive 
semidefinite, linear operators on $\B$ with unit trace:
$${\cS}({\cB})=\{\rho: \rho=\rho^*\geq 0,\tr\rho=1\},$$
where $\tr\rho=\sum_i \bra{i}\rho\ket{i}$ for some complete orthonormal 
basis $\{\ket{i}\}_i$. 

\begin{Definition}\label{DCQC}
A \emph{discrete classical-quantum channel (cq-channel)} is a
map $W:{\X}\longrightarrow {\cS}({\cB})$
where ${\X}$ is a finite set and ${\cS}(\B)$ the set of quantum
states of the complex Hilbert space $\B$, which we assume to be finite dimensional.
Furthermore, we denote $a=|\fset{X}|$ the cardinality of $\fset{X}$,
and $d=|\B|$ the dimension of $\B$. 
\end{Definition}

\medskip
Associated to $W$ is the channel map on a sequence of length $n$ over the alphabet $\X$.
$$W^{\otimes n}:{\X}^n\longrightarrow {\cS}(\B^{\otimes n})$$
with $W(x^n) = W^{\otimes n}(x^n) = W(x_1)\otimes\cdots\otimes W(x_n)$.
(Note that to abbreviate, we will customarily omit the superscript
${\otimes n}$ if the block length is evident from the input 
string $x^n$.) 
We call $W^{\otimes n}$ a memoryless channel.
In the following, we use the notation 
$W^{\otimes n}(P) \= \sum_{x^n\in A^n} P(x^n)W(x^n)$ to denote
the output state of the channel in $S(\B^{\otimes n})$ when 
the input is distributed according to $P$. 
To access the (classical) information of a quantum state, we have to perform
a measurement on the output space.

\begin{Definition}
Let $\B$ be a finite dimensional Hilbert space. A \emph{POVM 
(positive operator valued measure)} on $\B$  is a collection
$(D_i)_{i=1}^N$ of positive semidefinite operators $D_i$ on $\B$ such that
$\sum_{i=1}^N D_i = \1_{\B}$, where $\1_{\B}$ denotes
the identity operator on $\B$.
\end{Definition}

In transmission theory, Alice uses the classical-quantum channel 
to transmit messages from the set $\X$ to Bob. 
He tries to determine the transmitted messages by making a quantum
measurement (POVM). 

\begin{Definition}
An \emph{$(n,M,\lambda)$-code} is a set of pairs 
$\{(P_i,D_i):i\in [M]\=\{1,\dots,M\}\}$ where the $P_i$ are probability distributions on
$\cX^n$ and $D\=(D_i)_{i\in[M]}$ 
a POVM on ${\cB}^{\otimes n}$ such that: 
$\tr W^{\otimes n}(P_i)\!\cdot\!D_i  \ge 1-\lambda.$ 
The largest $M$ such that an $(n,\lambda)$-code exists
is denoted $M(n,\lambda)$.
\end{Definition}

The rate $R$ of a $(n, M,\lambda)$-code is defined as
$R = \frac{1}{n} \log M$. A rate $R$ is said to be \emph{achievable} 
if for all $\eta\in (0,1)$ 
there exists a $n_0(\eta)$, such 
that for all $n\geq n_0(\eta)$ there exists an $(n,2^{n(R-\eta)},\eta)$-code. 
The transmission capacity $C(\W)$ of a compound cq-channel $\W$ is the 
supremum of all achievable rates, which hence, is the largest achievable rate.
One of the main topics in quantum information theory is to determine the 
transmission capacities of channels. 

Let $\rho\in {\cS}(\A)$ be a state of a quantum system $\A$. We denote by 
$S(\rho^\A)=S(\A)=-\tr \rho^\A \log \rho^\A$ 
the von Neumann entropy. Furthermore, we define
the Holevo information $I(X:\B) = I(P;W) = S(W(P))-S(W|P)$ with
the output state $W(P)=\sum_{x\in{\X}} P(x)W(x)\in {\cS}(\B)$, 
and $S(W|P)=\sum_{x\in{\X}} P(x)S(W(x))$, the conditional entropy 
of the channel output for the input distribution $P$.

\begin{Theorem}[\cite{H98}, \cite{SW97}]
The classical transmission capacity of the cq-channel $W$, defined as
\[
  C(W) = \inf_{\lambda>0} \liminf_{n\rightarrow\infty} \frac1n \log M(n,\lambda),
\]
is given by
\begin{align*}
  \label{eq:C-cap1}
  C(W) = \max_{P(x)} I(P;W).
\end{align*}
Furthermore, the strong converse holds \cite{OgawaNagaoka:strong,winter:qstrong}:
\(
  \displaystyle\lim_{n\rightarrow\infty} \frac1n \log M(n,\lambda) = C(W).
\)
\qed
\end{Theorem}

\medskip
In \cite{Wilde:book}, more properties and results about 
transmitting classical information over quantum channels are discussed.

\subsection{Identification via cq-channels}

Compared to transmission, in identification theory we change the goal 
for Bob: We assume that he ``only'' wants to know if the transmitted message 
is equal to some $j$. 

\begin{Definition}
\label{def:ID-code}
An \emph{$(n,N,\lambda_1,\lambda_2)$ ID-code} is a set of pairs 
$\{(P_i,D_i):i\in [N]\}$ where the $P_i$ are probability distributions on
$\cX^n$ and the $D_i$ are POVM elements, i.e.~$0 \le D_i \le \1$, acting on
$B^{\otimes n}$, that $\forall\, i\neq j \in[N]$
\begin{align*}
  \tr W^{\otimes n}(P_i)\!\cdot\!D_i &\geq 1-\lambda_1 \text{ and }\\
  \tr W^{\otimes n}(P_i)\!\cdot\!D_j &\leq \lambda_2.
\end{align*}
The largest size $N$ of an $(n,N,\lambda_1,\lambda_2)$ ID-code
is denoted $N(n,\lambda_1,\lambda_2)$.
\end{Definition}

We use a stochastic encoder for the encoding of the 
messages; this is essential in the theory of identification.
The definition of an ID-code only partially fits the definition of a 
classical identification code in the following sense. There are applications
of classical identification codes, where one assumes that there are several
receivers, each only interested in one message, and all wanting to decide
individually if ``their'' message was sent. 
The example given in \cite{AhlswedeDueck:ID.A} is that of $N$ sailors on 
a ship, and each sailor is related to one relative.
On a stormy night, one sailor drowns in the ocean. One could now broadcast the name of the 
sailor to all relatives. However, this takes $\lceil \log_2 N\rceil$ bits. 
And the news is of interest only to one relative. If we now allow a certain error 
probability, we can broadcast an identification code using only
$O(\log_2\log_2 N)$ bits.

The ID-code for a quantum channel has the property that the received state 
cannot be used in general to ask for two different messages. The reason is
that the POVMs $(D_i,\1-D_i)$ are in general not compatible. Therefore the 
realisation of applications with more than one receiver, like in the 
example above, is not possible with an ID-code as defined. 
There are, however, applications where we have only two parties, who want to 
check if they have the same text (such as watermarking, or in the
communication complexity setting). 
L\"ober \cite{Loeber:PhD} defined simultaneous ID-codes to overcome this 
limitation. In this code model, there has to be one single measurement which 
allows us to identify every message at the same time.

\begin{Definition}
\label{SimuDef}
An ID-code $\{(P_i,D_i) : i\in[N]\}$ is called \emph{simultaneous}
if there is a POVM $(E_y)_{y\in\cY}$ acting on $B^{\otimes n}$
and subsets $\cA_i\subseteq \cY$,
such that $D_i = \sum_{y\in\cA_i} E_y$ for all $i\in[N]$.
The largest size of a simultaneous $(n,\lambda_1,\lambda_2)$ ID-code
is denoted $N_{\rm sim}(n,\lambda_1,\lambda_2)$.
\end{Definition}

In this case the measurement gives as a result some $y\in\cY$, 
and receiver $i$ has to check whether $y\in \cA_i$. Note that the
definition can be expressed equivalently by requiring that the measurements
$(D_1,\1-D_i)$ are all compatible, because this requires that there
exists a common refinement of them, i.e.~a POVM of which all 
$(D_1,\1-D_i)$ are coarse grainings.

\begin{rem}
\normalfont
If the $D_i$ are not compatible, there is no way of measuring them
all together jointly, but this does not mean that we have to give up.
To identify a set of messages $i_1,\ldots,i_k$, we could simply
apply the decoding POVMs $(D_{i_\kappa},\1-D_{i_\kappa})$ sequentially
in some order. That this is not a bad idea follows from the gentle
measurement lemma \cite{winter:qstrong}: since each measurement has
a high probability of giving the correct outcome, the state is 
disturbed, but only ``a little'' in trace norm, so we can subject the
next measurement as if nothing had happened at all.

The best analysis of this approach is using Sen's non-commutative
union bound \cite{Sen:noncommutative-union} in the version of Wilde
for general POVMs \cite{Wilde:sequential-dec}. Using this bound,
we can see that if we have any ID-code with errors $\lambda_1,\lambda_2 \leq \lambda$,
then we can correctly identify any set of $k \leq \frac{\epsilon^2}{4\lambda}$ 
messages, with error probability bounded by $\epsilon$. 
This will not include all messages, since for the rates below the capacity, the
error $\lambda$ can be made to vanish exponentially we get at least an
exponentially large $k$.
\end{rem}

\medskip
In the present paper we consider the identification capacity of a cq-channel,
of which we distinguish a priori simultaneous and non-simultaenous flavours,
following L\"ober \cite{Loeber:PhD}:

\begin{Definition}
  \label{defi:ID-capacities}
  The \emph{(simultaneous) classical ID-capacity} of a cq-channel $W$
  is defined as
  \begin{align*}
    C_{\rm ID}(W)           &\= \inf_{\lambda > 0} 
                                 \liminf_{n\rightarrow \infty} \frac1n \log\log N(n,\lambda,\lambda), \\
    C_{\rm ID}^{\rm sim}(W) &\= \inf_{\lambda > 0} 
                           \liminf_{n\rightarrow \infty} \frac1n \log\log N_{\rm sim}(n,\lambda,\lambda),
  \end{align*}
  respectively. 
\end{Definition}

%
%
%

L\"ober \cite{Loeber:PhD} showed that for cq-channels,
the simultaneous classical ID capacity is equal 
to the transmission capacity.
Furthermore, he showed that the strong converse holds for simultaneous ID-codes. 
Later, Ahlswede and Winter \cite{AhlswedeWinter:ID-q} extended the strong converse
to non-simultaneous ID-codes. 

\begin{Theorem}[\cite{Loeber:PhD}, \cite{AhlswedeWinter:ID-q}]
For any cq-channel $W$,
\begin{equation*}
  C_{\rm ID}^{\rm sim}(W) = C_{\rm ID}(W) = C(W),
\end{equation*}
and the strong converse holds: for all $\lambda_1+\lambda_2 < 1$,
\begin{align*}
  \phantom{.}
  \lim_{n\rightarrow \infty} &\frac{1}{n} \log\log N(n,\lambda_1,\lambda_2)\\ 
  &=  \lim_{n\rightarrow \infty} \frac{1}{n} \log\log N_{\rm sim}(n,\lambda_1,\lambda_2)
   = C(W).
  \hspace{3mm}\blacksquare
\end{align*}
\end{Theorem}

Ahlswede and Winter also considered the case of a general (quantum-quantum)
channel, but the results are much less complete. It is not even clear
if in the general case the simultaneous capacity and the non-simultaneous 
ID-capacity coincide. See the subsequent papers \cite{HW12} and the
review \cite{winter:survey} for a presentation of the state of the art.

Definitions \ref{def:ID-code} and \ref{SimuDef} address Freeman Dyson's 
critique on the status quo of experiments, measurements, and detectors in 
particle physics (of course in our setting of operational tasks).  
According to Dyson, experiments as currently conducted in particle physics,
can only answer very specific questions. Analogous to our model, this corresponds 
to identification codes (Definition  \ref{def:ID-code}), and in particular the use of
\emph{message-dependent} measurements. 
In comparison, the simultaneous identification codes provide universal 
measurements so that the relevant questions can be answered by classical 
post-processing. 

We will also show that the same performance can be achieved with this code 
concept as with fully general identification codes.  
Of course, particle physics applications do not have the luxury€ of being 
able to control the encoding in general.

\section{Identification via robust cq-channels}
\label{Robust}
In this section we will define the identification capacity of a compound cq-channel
and derive its single-letter formula. 
In \cite{M15} and \cite{BJK17}, the transmission capacity was derived.
We will use the transmission code and build an identification code
with the method introduced in \cite{AhlswedeDueck:ID.B}. 
This method was also used in \cite{Loeber:PhD} to get the identification capacity 
of a cq-channel. For the converse we generalize the method of
\cite{AhlswedeWinter:ID-q}. 

\begin{Definition}
\label{CQCC}
Let $\Theta$ be an index set, 
$\X$ a finite set and $\B$ a finite-dimensional Hilbert space. 
Let $W_t:\X \longrightarrow \cS(\B)$ be a cq-channel for every $t\in \Theta$:
\[
  W_t:\X\ni x\mapsto W_{t}(x)\in \mathcal{S}(\B).
\]
The memoryless extension of the cq-channel $W_t$ is given by 
$W_t(x^n) = W_t^{\otimes n}(x^n) = W_{t}(x_1)\otimes \ldots \otimes W_{t}(x_n)$ for $x^n \in \X^n$.
We then call $\W \= \{W_t\}_{t\in \Theta}$ a \emph{compound cq-channel}.
\end{Definition}

\begin{Definition}
An \emph{$(n,M,\lambda)$-code} for the compound cq-channel $\W$ is a family 
$\mathcal{C}\=\left( (P_m,D_m) : m\in[M] \right)$ consisting of pairs of stochastic
encodings given by code word probability distributions $P_m$ over $\X^n$ 
and positive semi-definite operators  $D_i\in \mathcal{B}(B^{\otimes n})$ 
forming a POVM, i.e.~$\sum_{m=1}^{M} D_m = \1_{B^n}$, such that
\begin{equation*}
   \sup_{t\in\Theta} \max_{i\in[M]} 1-\tr W_{t}^{\otimes n}(P_i) D_i  
    \le \lambda.
\end{equation*}

The number $M$ is called the \emph{size} of the code, and $\lambda$ the error 
probability. The maximum $M$ for given $n$ and $\lambda$ is
denoted $M(n,\lambda)$, extending the definition for a cq-channel
(which is recovered for $|\Theta|=1$).

The capacity of $\W$ is defined as before,
\[
  C(\W) = \inf_{\lambda>0} \liminf_{n\rightarrow\infty} \frac1n \log M(n,\lambda).
\]
\end{Definition}

Thus an $(n,M,\lambda)$-code for the compound cq-channel $\W$ ensures that the maximal 
error probability for all channels $W_t$ is uniformly bounded above by 
$\lambda$. A more intuitive description of the compound cq-channel is that 
the sender and receiver do not know which channel from the set $\W$ 
is actually used during the transmission of the $n$-block; their prior knowledge is 
merely that the channel is memoryless and belongs to the set $\W$. Their task
is to prepare for the worst case among those.


\begin{Theorem}[\cite{BB09}]
  \label{CCQ-cap}
  Let $\W$ be a compound cq-channel with finite input alphabet $\X$ 
  and finite-dimensional output Hilbert space $\B$. Then,
  $\displaystyle{C(\W) = \max_{P(x)} \inf_{t\in\Theta} I(P;W_t).}$
  \qed
\end{Theorem}

We stress that we explicitly allow stochastic encoders in the definition. 
It is known that this does not change the capacity compared to deterministic
encoders, although it makes it easier for us
to relate later channel models to compound cq-channel coding results.
Note however, that this change implies that the average error probability
criterion and the maximum error criterion lead to the same achievable 
rates, and so the strong converse does not hold any more, only the weak
converse.

\begin{Definition}
An \emph{$(n,N,\lambda_1,\lambda_2)$ ID-code} for the compound cq-channel
$\W$ is a set of pairs $\{(P_i,D_i):i\in[N]\}$, where the $P_i$ are 
probability distributions on $\X^n$ and the $D_i$ are POVM elements, 
i.e.~$0 \le D_i\le \1$ acting on $\B^{\otimes n}$, such that $\forall\ i\neq j\in[N]$
\begin{align*}
 \inf_{t\in\Theta} \tr W_t^{\otimes n}(P_i)\!\cdot\!D_i &\geq 1-\lambda_1 
    \text{ and }\\
   \ \sup_{t\in\Theta} \tr W_t^{\otimes n}(P_i)\!\cdot\!D_j &\leq \lambda_2.
\end{align*}
The largest size of an $(n,N,\lambda_1,\lambda_2)$ ID-code
is denoted $N(n,\lambda_1,\lambda_2)$. Analogous to previous
definitions, we also have simultaneous ID-codes and the
maximum code size $N_{\rm sim}(n,\lambda_1,\lambda_2)$. 
\end{Definition}

The identification capacities are defined as before. All capacities 
in this paper are defined in the so-called pessimistic way. 
The optimistic definition of capacity is 
$\bar{C}(\W) = \inf_{\lambda>0} \limsup_{n\rightarrow\infty} \frac1n \log M(n,\lambda)$. 
We show that the converse holds for the optimistic definition and therefore 
also for the pessimistic definition.

\begin{Theorem}\label{IDcompound}
  \label{cqid}
  Let $\W$ be an arbitrary compound cq-channel with finite input alphabet 
  $\X$ and finite-dimensional output Hilbert space $\B$. Then,
  \begin{equation*}
    C_{\rm ID}(\W) = C_{\rm ID}^{\rm sim}(\W) = C(\W)= \inf_{t\in\Theta} I(Q;W_t),
  \end{equation*}
  and the weak converse holds for the optimistic ID-capacity. Indeed,
  \[\begin{split}
    \inf_{\lambda_1,\lambda_2>0}
       &\limsup_{n\rightarrow\infty} \frac1n \log N(n,\lambda_1,\lambda_2) \\
       &= \inf_{\lambda_1,\lambda_2>0}
         \liminf_{n\rightarrow\infty} \frac1n \log N_{\rm sim}(n,\lambda_1,\lambda_2) \\
       &= C(\W).
  \end{split}\]
\end{Theorem}

\begin{proof}
We will use an $(n,M,\lambda)$-code for the compound channel to construct an 
$(n,N,\lambda_1,\lambda_2)$ ID-code.
We use the following lemma from \cite{Loeber:PhD}, which is a slightly 
modified version of the original in \cite{AhlswedeDueck:ID.A}:

\begin{Lemma}[\cite{Loeber:PhD}]
\label{GilbLem}
Let ${\M}$ be a finite set of cardinality $M$ and let $\lambda\in(0,1)$.
Let $\varepsilon>0$ be small enough so that 
$\lambda\log_2(\frac{1}{\varepsilon}-1) > 2$. Then there are at least 
$N \ge \frac{1}{M}2^{\lfloor\varepsilon M\rfloor}$ subsets 
${\A}_1,\ldots,{\A}_N\subset {\M}$, each of cardinality
$\lfloor\varepsilon M\rfloor$, such that the cardinalities of the
pairwise intersections satisfy, for all  $i\neq j\in [N]$,
\[
  \hspace{2.3cm}
        |{\A}_i \cap {\A}_j| < \lambda\lfloor\varepsilon M\rfloor .
  \hspace{2.3cm}
  \blacksquare
\]
\end{Lemma}

By definition of the transmission capacity, there is an $(n,M,\lambda)$-code
${\cal C}=\{(c_m,E_m) : m \in [M] \}$ with $M\ge 2^{(C(\W)-\delta)n}$ if
$n$ is large enough. Using $[M]$ as ground set, 
Lemma \ref{GilbLem} provides us with subsets 
${\cal A}_1,\ldots,{\cal A}_N\subset [M]$ of cardinality
$\lfloor\varepsilon M\rfloor$ with pairwise intersections smaller than
$\lambda\lfloor\varepsilon M\rfloor$. Here we have for the number $N$ of 
those sets:
\[
  N \ge \frac{1}{M}2^{\lfloor\varepsilon M\rfloor}
             \underset{n\gg 1}\ge 2^{\,\lfloor\varepsilon 2^{(C(\W)-\delta)n}\rfloor - n}.
\]
We construct a simultaneous ID-code $\{(P_i,D_i):i=1,\ldots,N\}$ by taking 
as $P_i$ the uniform distribution on sets ${\cal C}_i \= \{c_m : m\in{\cal A}_i\}$,
and as $D_i$ the sum of the corresponding $E_m$'s:
\begin{align*}
  P_i(x^n) &\= 
  \begin{cases}
    \frac{1}{\lfloor\varepsilon M\rfloor} & \text{if } x^n\in{\cal C}_i, \\
    \hspace{9pt} 0                        & \text{otherwise},
  \end{cases} \\
  \text{and } D_i &\= \sum_{m\in{\cal A}_i} E_m \qquad (i=1,\ldots,N).
\end{align*}
We choose $\lambda_1\geq \lambda$ and $\lambda_2\geq 2\lambda$. 
It is now straightforward to bound the errors:
\begin{align*}
\min_{t\in\Theta}& \tr W_t^{\otimes n}(P_i) \!\cdot\! D_j\\
  &=    \frac{1}{\lfloor\varepsilon M\rfloor}     
            \min_{t\in\Theta}\sum_{m\in{\cal A}_i}\sum_{m'\in{\cal A}_i} \tr W^n_{t,c_{m'}} E_m\\
  &\geq \frac{1}{\lfloor\varepsilon M\rfloor}
            \min_{t\in\Theta}\sum_{m\in{\cal A}_i} \tr W^n_{t}(c_m) E_m\\
  &\geq 1-\lambda \geq 1-\lambda_1.
\end{align*}
For $i\neq j$,
\begin{align*}
  \max_{t\in\Theta}& \tr W_t^n(P_i) \!\cdot\! D_j\\
  &= \frac{1}{\lfloor\varepsilon M\rfloor}
        \max_{t\in\Theta} \sum_{m\in{\cal A}_j}\sum_{m'\in{\cal A}_i} \tr W^n_{t,c_{m'}} E_m \\
  &= \frac{1}{\lfloor\varepsilon M\rfloor} \max_{t\in\Theta} \sum_{m\in{\cal A}_j} 
        \left(\sum_{m'\in{\cal A}_i\cap{\cal A}_j} \tr W^n_{t,c_{m'}} E_m \right. \\
  &\phantom{===========}
        \left. + \sum_{m'\in{\cal A}_i\setminus{\cal A}_j} \tr W^n_{t,c_{m'}} E_m \right) \\
  &\leq \frac{1}{\lfloor\varepsilon M\rfloor} 
         \left( \! \sum_{m'\in{\cal A}_i\cap{\cal A}_j}
         \!\underbrace{\max_{t\in\Theta} \tr W^n_{t,c_{m'}}
                              \!\left(\sum_{m\in{\cal A}_j}E_m\right)}_{\leq\,1} \right. \\
  &\phantom{====:}
         \left. + \sum_{m'\in{\cal A}_i\backslash{\cal A}_j}
         \!\underbrace{\max_{t\in\Theta} \tr W^n_{t,c_{m'}}
                        \!\left(\sum_{m\in{\cal A}_j}E_m\right)}_{\leq\,\lambda} \! \right)  \\
  &\leq \frac{1}{\lfloor\varepsilon M\rfloor} 
          \lambda\lfloor\varepsilon M\rfloor 
          + \frac{1}{\lfloor\varepsilon M\rfloor}
            \lfloor\varepsilon M\rfloor \lambda
   =    2\lambda \leq \lambda_2,
\end{align*}
where we have used $|{\A}_i \cap {\A}_j| < \lambda\lfloor\varepsilon M\rfloor$.
Therefore we have shown

\begin{equation*}
  C_{\rm ID}(\W)\geq C_{\rm ID}^{\rm sim}(\W) \geq C(\W).
\end{equation*}

It remains to prove the converse, i.e.
\begin{equation*} 
  C_{\rm ID}^{\rm sim}(\W) \leq C_{\rm ID}(\W) \leq C(\W).
\end{equation*}
For this, consider an arbitrary (non-simultaneous) $(n,\lambda_1,\lambda_2)$-ID 
code $\{(P_i,D_i) : i\in [N] \}$.

1. The first step follows by applying the theory of types. Fix a $\delta$-net $\cT \subset \cP(\X)$
on the probability distributions on $\X$, i.e.~for any p.d.~$P$ there exists
a $Q \in \cT$ with $\frac12 \|P-Q\|_1 \leq \delta$. It is known that such a 
net can be chosen with $|\cT| \leq \left(\frac{c}{\delta}\right)^{|\X|}$
for a constant $c>0$. 
This induces a partition of the input space
$\X^n = \bigcup_{Q\in\cT} \cA_Q$ in such a way that the type (i.e.~the
empirical distribution) of each $x^n \in \cA_Q$ is $\delta$-close to $Q$.
Now we can write each distribution $P_i$ as 
\[
  P_i = \bigoplus_{Q\in\cT} \mu_i(Q) P_{iQ},
\]
where $\mu_i$ is a p.d.~over $\cT$ and the $P_{iQ}$ are p.d.'s over
$\cA_Q$ (extended trivially to all of $\X^n$). Choose an $\epsilon$-net 
$\cM \subset \cP(\cT)$ which can be found with 
$|\cM| \leq \left(\frac{c}{\epsilon}\right)^{|\cT|}$. Then there exists a
$\overline{\mu}\in\cM$ such that at least a fraction of $\frac{1}{|\cM|}$
of the messages has its $\mu_i$ $\epsilon$-close to $\overline{\mu}$: w.l.o.g.,
\[
  \forall i\in [N']=\left\lfloor \frac{N}{|\cM|} \right\rfloor \quad 
                             \frac12 \|\mu_i-\overline{\mu}\|_1 \leq \epsilon.
\]
Now modify the code as follows:
\[
  P_i' \= \bigoplus_{Q\in\cT} \overline{\mu}(Q) P_{iQ},
\]
for $i\in [N']$, leaving the $D_i$ unchanged.
This clearly gives an $(n,N',\lambda_1',\lambda_2')$-ID
code with $\lambda_1'=\lambda_1+\epsilon$, $\lambda_2'=\lambda_2+\epsilon$.
If $\lambda_1+\lambda_2 < 1$, we can choose $\epsilon$ small
enough to ensure that $\lambda_1'+\lambda_2' < 1$.

2. Now, there exists a $Q \in \cT$ with $\overline{\mu}(Q) \geq \frac{1}{|\cT|}$.
Modify the code once more by truncating all other contributions $Q'\in\cT$,
i.e.~consider the code $\{(P_i''\=P_{iQ},D_i) : i\in [N'] \}$. Since the encodings
are a small, but not too-small fraction of the $P_i'$, the error probabilities
can increase significantly, but we can control them. Concretely, the new code has
errors of the first and second kind, $\lambda_1'' = |\cT|\lambda_1'$ and
$\lambda_2'' = |\cT|\lambda_2'$, respectively. Since we are in the weak
converse regime of $\lambda_1+\lambda_2 \rightarrow 0$ asymptotically, 
we can choose  $\epsilon \rightarrow 0$ sufficiently slowly so that
$\lambda_1'+\lambda_2' \rightarrow 0$ too, and hence for each $\delta > 0$,
$\lambda_1''+\lambda_2'' \rightarrow 0$. As we selected a fraction of the 
messages that is going to zero arbitrarily slowly, we have the same asymptotic 
rate.

3. At this point we are in a good position: all the code distributions
$P_i''=P_{iQ}$ are supported on $\cA_Q$, which is a subset of the 
$\delta$-typical sequences (in the sense of frequency typicality).
Which means we can apply the converse proof from \cite{AhlswedeWinter:ID-q}
to each $W_t\in\W$, $t\in\Theta$, obtaining
\[
  \frac1n \log\log N' \leq I(Q;W_t) + O\bigl(\delta\log|B| + h_2(\delta)\bigr),
\]
the latter terms occurring because the types of sequences in $\cA_Q$ 
fluctuate up to $\delta$ around $Q$. This completes it, since we can choose
$\delta>0$ arbitrarily small, and as explained above, $\epsilon$ can be made
to go to $0$ arbitrarily slowly. Hence,
\begin{align*}
  \frac1n \log\log N &\leq \frac1n \log\log N' + o(1)\\
                     &\leq \inf_{t\in\Theta} I(Q;W_t) + o(1),
\end{align*}
concluding the proof.
\end{proof}

\section{Secure identification via \protect\\ wiretap cqq-channels}
\label{Secure}
An important aspect in information theory is security, or privacy. 
Wyner \cite{W75} introduced the classical wiretap channel, which he solved
in the degraded case, and later Csisz\'ar and K\"orner \cite{CK78}
in the general case.
It can be described by two channels from the sender (``Alice'') 
to the legal receiver (``Bob'') and to the eavesdropper (``Eve''), respectively. 
In transmission theory the goal is to send messages to the legal receiver, 
while the wiretapper is to be kept ignorant. The wiretap channel was generalized
to the setting of quantum information theory in \cite{CWY04,D05}. 
Formally, in contrast to the classical case, quantumly the channel has to be 
described by a single quantum operation $T$, from Alice to the joint 
system of Bob and Eve together: then we can define the legal channel
$W = \tr_B \circ T$ and the wiretapper channel $V = \tr_E \circ T$.
Note that (unlike the classical case) this pair
of channels cannot be arbitrary! This has to do with the no-cloning theorem: 
Alice's input state cannot be duplicated and then sent through both channels. 

However, here we will restrict ourselves to the cq-channel case, where
Alice's input is described by a letter $x\in\X$ from a finite alphabet.
Then we can define the classical-quantum wiretap channel in a simple way.

\begin{Definition}
A \emph{classical-quantum wiretap channel (wiretap cqq-channel)} is a pair 
$(W,V)$ of two discrete memoryless cq-channels $W:\cX \longrightarrow \cS(\cB)$ 
and $V:\cX \longrightarrow \cS(\cE)$. 
When Alice sends a classical input $x^n\in{\X}^n$, Bob (legal receiver) and Eve 
(eavesdropper) receive the states $W^{\otimes n}(x^n)$ and $V^{\otimes n}(x^n)$, 
respectively. 
\end{Definition}

\begin{Definition}
An \emph{$(n,M,\lambda,\mu)$-wiretap code} for the wiretap cqq-channel 
$(W,V)$ is a collection $\{(P_i,D_i):i\in [M]\}$ of pairs consisting
of probability distributions $P_i$ on $\cX^n$ and a POVM $(D_i)_{i=1}^N$
on $\cB^{\otimes n}$ such that
\begin{align*}
  \forall i\in[M]\quad   & 1-\tr W^{\otimes n}(P_i)\!\cdot\!D_i \leq \lambda, \\
  \forall i,j\in[M]\quad & \frac12 \| V^{\otimes n}(P_i) - V^{\otimes n}(P_j) \|_1 \leq \mu.
\end{align*}

The largest $M$ such that an $(n,M,\lambda,\mu)$-wiretap code exists
is denoted $M(n,\lambda,\mu)$. The \emph{secrecy capacity} 
(aka \emph{private capacity}) of $(W,V)$ is then defined as
\[
  C_S(W,V) \= \inf_{\lambda,\mu>0} \liminf_{n\rightarrow\infty} \frac1n \log M(n,\lambda,\mu).
\]
\end{Definition}

Note that by the Fannes inequality \cite{Fannes,AW:S}, the second
condition (``privacy'') implies that for any random variable $J$ taking
values in $[M]$, $I(J:E^n) \leq \mu n \log|E| + h(\mu)$. It turns out that
the right hand side can be made arbitrarily small while achieving the capacity, 
because $\mu$ as well as $\lambda$ can be made to converge to $0$ to
any polynomial order.

\begin{Theorem}[\cite{CWY04}]
\label{QWC}
The secrecy capacity of a wiretap cqq-channel is given by
\begin{align*}
  C_{S}&(W,V) \\
       &= \lim_{n\rightarrow \infty} 
                 \max_{U\rightarrow X^n \rightarrow B^n E^n} 
                  \frac{1}{n}\bigl( I(U:B^{n}) - I(U:E^{n}) \bigr),
\end{align*}
where the maximum is taken over all random variables that satisfy
the Markov chain relationships $U\rightarrow X^n \rightarrow B^n E^n$. \qed
\end{Theorem}

Thus in the
case of transmission theory, we have a positive secrecy
capacity $C_S$ when the channel parameters of the legal
channel are ``better'' than those of the non-legal channel.
This means we pay a price in the form of a smaller rate
for secure transmission. We will show that in the case of
identification, the situation is different.

\begin{Definition}\label{wID}
An \emph{$(n,N,\lambda_1,\lambda_2,\mu)$ wiretap ID-code} for the 
wiretap cqq-channel $(W,V)$ is a set of pairs 
$\{(P_i,D_i):i\in[N]\}$, where the $P_i$ are probability distributions on
$\X^n$, and the $D_i$, $0 \le D_i \leq \1$ denote operators on 
$\B^{\otimes n}$, such that for all $i\neq j\in[N]$ and for all $0\leq F\leq \1_{E^n}$,
\begin{align}
  \tr W^{\otimes n}(Q_i) D_i &\ge 1-\lambda_1, \nonumber \\ 
  \tr W^{\otimes n}(Q_j) D_i &\le \lambda_2,   \nonumber \\
  \tr V^{\otimes n}(Q_j) F + \tr V^{\otimes n}(Q_i) (\1-F) &\geq 1-\mu.\label{ID-secrecy}
\end{align}
If the POVMs $(D_i,\1-D_i)$ are all compatible, we call the code
\emph{simultaneous}, as in the cq-channel case.
\end{Definition}
 

Condition (\ref{ID-secrecy}) enforces that the wiretapper cannot very well distinguish 
the output states $Q_iV^{\otimes n}$ of the different
messages. Indeed, it is equivalent to
\begin{align*}
  \mu &\geq \max_{0\leq F\leq \1} \tr (V^{\otimes n}(Q_j)-V^{\otimes n}(Q_i)) F\\
      &=    \frac12 \| V^{\otimes n}(Q_j)-V^{\otimes n}(Q_i) \|_1,
\end{align*}
which by Helstrom's theorem \cite{Helstrom,NielsenChuang} means that
even if Eve somehow knows that the message can only be either $i$ or $j$
with equal probability, then her error probability for discriminating 
these two alternatives is at least $\frac12(1-\mu) \approx \frac12$.

The maximum $N$ for which a 
$(n,N,\lambda_1,\lambda_2,\mu)$ wiretap ID-code exists is denoted by
$N(n,\lambda_1,\lambda_2,\mu)$. For simultaneous wiretap ID-codes
we denote the maximum $N_{\rm sim}(n,\lambda_1,\lambda_2,\mu)$.
We then define the (simultaneous) secure identification capacity of the 
wiretap channel as
\begin{align*}
C_{\rm SID}&(W,V)\\
           &\= \inf_{\lambda_1,\lambda_2,\mu>0} 
                 \liminf_{n\rightarrow\infty} \frac1n \log\log N(n,\lambda_1,\lambda_2,\mu),\\
C_{\rm SID}^{\rm sim}&(W,V)\\ 
           &\= \inf_{\lambda_1,\lambda_2,\mu>0} 
                  \liminf_{n\rightarrow\infty} \frac1n \log\log N_{\rm sim}(n,\lambda_1,\lambda_2,\mu),
\end{align*}
respectively.

In this section we consider the wiretap cqq-channel and derive a multi-letter
formula for its secure identification capacity. 
The idea is similar to the classical case. We use a combination of two codes. 
For the converse we generalize inequalities of \cite{AZ95} and \cite{FG99}. 

\begin{Theorem}[Dichotomy theorem]
\label{Theorem D2}
Let $C(W)$ be the capacity of the cq-channel $W$ and let $C_{S}(W,V)$ be
the secrecy capacity of the wiretap cqq-channel. Then,
\begin{align*}
  C_{\text SID}(W,V) &= C^{\rm sim}_{\text SID}(W,V) \\
          &= \begin{cases}
               C(W) & \text{ if } C_S(W,V) > 0,\\
               0    & \text{ if } C_S(W,V) = 0.
             \end{cases}
\end{align*}
\end{Theorem}

\begin{proof}
For the direct part, the identification code is constructed by means 
of two fundamental codes, following \cite{AhlswedeDueck:ID.B}.

Let $0<\epsilon<C$ be fixed. We know that there is a 
$\delta>0$ such that for sufficiently large $n$ there is an
$(n,M',\lambda(n))$-code $C'=\left\{\left(u'_j,\D'_j | j\in [M']\right)\right\}$ for the cq-channel
with code size $M'=\lceil 2^{n(C(\W)-\epsilon)}\rceil$
and by Theorem~\ref{QWC} an $(\lceil\sqrt{n}\,\rceil,M'',\lambda(\sqrt{n}),\mu(\sqrt{n}))$ wiretap code 
$C''=\left\{\left(u''_k,\D'_k | k\in [M'']\right)\right\}$ for the wiretap cqq-channel $(W; V )$ with code size 
$M''=\lceil2^{\epsilon\sqrt{n}} \rceil$.

In Theorem~\ref{QCWCdich}, the construction for the more general compound model is described. 
We use the same idea here to show the direct part: Alice and Bob first create shared randomness 
with the help of the code $C'$ at a rate equal to the channel capacity. 
A code with an arbitrary small positive rate is then sufficient to use the method of 
Ahlswede and Dueck by sending and decoding the function values.
For this purpose we use $C''$.

Furthermore, we have to show that if $C_{SID}(W,V)>0$, then $C_S(W,V)>0$.
We begin with the following two lemmas. To state them, 
we fix two messages $i$ and $j$ in an $(n,N,\lambda_1,\lambda_2,\mu)$ wiretap ID-code,
and consider the uniform distribution $Q$ on $\{i,j\}$.

We define a new wiretap channel $(\widetilde{W},\widetilde{V})$, which has binary 
input $\{i,j\}$ and output states in $\cS(\cB^n)$ and $\cS(\cE^n)$, respectively:
it acts by mapping $i$ and $j$ to the input distribution $P_i$ and $P_j$
on $\cX^n$, respectively, on which the wiretap cqq-channel $(W^{\otimes n},V^{\otimes n})$
operates then, yielding outputs 
\begin{align*}
  \widetilde{W}_i &= W^{\otimes n}(P_i),\ \widetilde{W}_j = W^{\otimes n}(P_j), \\
  \widetilde{V}_i &= V^{\otimes n}(P_i),\ \widetilde{V}_j = V^{\otimes n}(P_j).
\end{align*}

\begin{Lemma}
  \label{Lemma 3} 
  If for any POVM element $0\leq F\leq \1$ on $\cE^{\otimes n}$,
  \[
    \tr(V^{\otimes n}(P_j) F) + \tr(V^{\otimes n}(P_i) (\1-F)) \geq 1-\mu, 
  \]
  then with $Q$ the uniform distribution on $\{i,j\}$,
  \[
    I(Q;\widetilde{V}) \leq h\left(\frac{\mu}{2}\right).
  \]
\end{Lemma}

\begin{proof}
As remarked after Definition~\ref{wID}, the condition on $F$ means that
$\frac12 \| \widetilde{V}_i-\widetilde{V}_j \|_1 \leq \mu$.
By the inequalities relating trace norm and fidelity \cite{FG99},
this means that $F(\widetilde{V}_i,\widetilde{V}_j) \geq 1-\mu$.
Invoking furthermore Uhlmann's theorem \cite{Uhlmann,Jozsa:fidelity},
we know that there exist purifications $\ket{\varphi_i}$ and
$\ket{\varphi_j} \in \cB^{\otimes n} \otimes \cC$
of $\widetilde{V}_i$ and $\widetilde{V}_j$, respectively,
\[
  \tr_C \varphi_x = \widetilde{V}_x,\ x\in\{i,j\},
\]
such that
\[
  F(\varphi_i,\varphi_j) = F(\widetilde{V}_i,\widetilde{V}_j) \geq 1-\mu.
\]

Hence, by the data processing inequality,
\[
  I(Q;\widetilde{V}) \leq I(Q;\varphi),
\]
where we interpret the two states $\varphi_i$ and $\varphi_j$ as
a binary cq-channel. To get the desired upper bound, we need to 
maximise the right hand side above over all pairs of pure states
$\varphi_i$ and $\varphi_j$, with
$F(\varphi_i,\varphi_j) = |\bra{\varphi_i}\varphi_j\rangle| \geq 1-\mu$.
This can be done explicitly because it is effectively a 
two-dimensional problem in the span of the two state vectors.
Indeed, due to unitary invariance, we may w.l.o.g.~write
\begin{align*}
  \ket{\varphi_i} = \alpha\ket{0}+\beta\ket{1}, \\
  \ket{\varphi_i} = \alpha\ket{0}-\beta\ket{1},
\end{align*}
where $\alpha \geq \beta \geq 0$ are real and non-negative and 
$\alpha^2+\beta^2=1$. The fidelity constraint means that
$|\alpha^2-\beta^2|\geq 1-\mu$, which translates into 
$2\beta^2 \leq \mu$. On the other hand, the Holevo information
reduces to $I(Q;\varphi) = S\bigl( \alpha^2\proj{0}+\beta^2\proj{1}\bigr) = h(\beta^2)$,
which can be at most $h\left(\frac{\mu}{2}\right)$, as claimed.
\end{proof}

\medskip
\begin{Lemma}
  \label{Lemma 4}
  If for some POVM element $0\leq D\leq \1$ (for instance the
  decoding POVM element $D_i$ from the ID-code),
  $\tr(W^{\otimes n}(P_i) D) \ge 1-\lambda_1$
  and 
  $\tr(W^{\otimes n}(P_j) D) \le \lambda_2$,
  with $\lambda_1,\lambda_2 \leq \frac12$, then
  \[
    I(Q;\widetilde{W}) \geq h\left(\frac12(1+\lambda_1-\lambda_2)\right)
                             - \frac12 h(\lambda_1) - \frac12 h(\lambda_2).
  \]
\end{Lemma}

\begin{proof}
We construct a binary channel with inputs and outputs $\{i,j\}$,
by performing the binary measurement $(D,\1-D)$ on the states
$\widetilde{W}_x$, $x\in\{i,j\}$, leading to an output $y\in\{i,j\}$
and thus defining a channel $T$ via
\begin{align*}
  T(i|x) &= \tr(W^{\otimes n}(P_x) D),\\
  T(j|x) &= 1-\tr(W^{\otimes n}(P_x) D).
\end{align*}
By data processing (in fact, the original Holevo bound!), we have
\begin{align*}
  I(Q;\widetilde{W})
    &\geq I(Q;T) \\
    &\geq h\!\left(\! \frac12(1+\lambda_1-\lambda_2) \!\right)
          - \frac12 h(\lambda_1) - \frac12 h(\lambda_2),
\end{align*}
the last by an elementary calculation.
\end{proof}

\medskip
Returning to the converse proof, recall that
the existence of identification codes at a positive rate implies
the following for the messages:
\begin{enumerate}
\item For all $i\in [N]$,
      $\tr(W^{\otimes n}(Q_i) D_i) \ge 1-\lambda$;
\item for all $i,j\in [N]$ with $i\neq j$,
      $\tr(W^{\otimes n}(Q_j) D_i) \le \lambda$;
\item for all $i,j\in [N]$ with $i\neq j$ and any operator $F$ on 
${\cB}^{\otimes n}$:
\[
  \tr(V^{\otimes n}(Q_j) F) + \tr(V^{\otimes n}(Q_i) (\1-F)) \geq 1-\lambda,
\]
\end{enumerate}
where $\lambda_1,\lambda_2,\mu \leq \lambda \leq \frac12$.

From the first two properties, it follows by Lemma~\ref{Lemma 4} that 
\[
  I(Q;W)\geq 1-h(\lambda).
\]
By the third property and Lemma~\ref{Lemma 3},
\[ 
  I(Q;V) \leq h\left(\frac{\lambda}{2}\right) \leq h(\lambda).
\]
Thus if $2h(\lambda) \leq 1$, which is true for all $\lambda \leq \frac{1}{15}$, 
we obtain
$I(Q;\widetilde{W})>I(Q;\widetilde{V})$ and so
$C_S(\widetilde{W},\widetilde{V})>0$, which therefore must hold for the
original channel $(W,V)$ as well.
\end{proof}

\medskip
\begin{rem}
In the classical case, or more generally when $\widetilde{V}_i$ and
$\widetilde{V}_j$ commute, the upper bound of Lemma~\ref{Lemma 3}
can be improved to
\[
  I(Q;\widetilde{V}) \leq \mu,
\]
cf.~\cite{FG99}.
To see this, we use a well-known characterisation of the total
variational distance (the commutative trace distance):
\begin{align*}
  \frac12 \| \widetilde{V}_i-\widetilde{V}_j \|_1 
    &= \min t \text{ s.t. } 
              \exists\, \widetilde{V}_i^0, \widetilde{V}_j^0, \widetilde{V}_\perp^0 
              \text{ states with}\\
    &\phantom{========}             
              \widetilde{V}_i = t\widetilde{V}_i^0 + (1-t)\widetilde{V}_\perp^0,\\
    &\phantom{========}             
              \widetilde{V}_j = t\widetilde{V}_j^0 + (1-t)\widetilde{V}_\perp^0.
\end{align*}
The optimal $(1-t)\widetilde{V}_\perp^0$ is simply 
$\min\bigl(\widetilde{V}_i,\widetilde{V}_j\bigr)$, which for classical 
discrete probability distributions is defined pointwise, and for
commuting density operators via functional calculus. In particular, 
we can choose states $\widetilde{V}_i^0$, $\widetilde{V}_j^0$
and $\widetilde{V}_\perp^0$, such that 
\[
  \widetilde{V}_i = \mu\widetilde{V}_i^0 + (1-\mu)\widetilde{V}_\perp^0, \quad 
  \widetilde{V}_j = \mu\widetilde{V}_j^0 + (1-\mu)\widetilde{V}_\perp^0.
\]
This can be interpreted as a factorisation of the channel 
$\widetilde{V} = \widetilde{V}^0 \circ \cE_{\mu}$
into an erasure channel $\cE_{\mu}:\{i,j\} \longrightarrow \{i,j,\perp\}$,
with $\cE_\mu(x|x) = \mu$ and $\cE_\mu(\perp|x)=1-\mu$,
and the cq-channel $\widetilde{V}^0:\{i,j,\perp\} \longrightarrow \cS(E^n)$
defined by the states $\widetilde{V}_i^0$, $\widetilde{V}_j^0$
and $\widetilde{V}_\perp^0$. By data processing,
\[
  I(Q;\widetilde{V}) \leq I(Q;\cE_\mu) = \mu,
\]
and we are done.
\qed
\end{rem}

\section{Secure identification via \protect\\ robust wiretap cqq-channels}
\label{Securerobust}
In this section we consider robust and secure cq-channels. 
The results for transmission capacities can be found in \cite{BCCD14}
and \cite{BCDN17}: In \cite{BCCD14} the secrecy of the classical compound 
channel with quantum wiretapper and channel state information (CSI)
at the transmitter was derived. 
Furthermore, a lower bound on the secrecy capacity of this channel without 
CSI and the secrecy capacity of the compound 
classical-quantum wiretap channel with CSI at the transmitter is determined. 
In \cite{M15}, a multi-letter formula for the secrecy capacity of the compound 
classical-quantum wiretap channel is given.
We will show that the capacity of a compound 
wiretap cqq-channel again satisfies a dichotomy theorem.

\begin{Definition}
Let $\Theta$ and $\Sigma$ be index sets and let $\W=\{W_t:\X\rightarrow\cS(\cB):t\in\Theta\}$ 
and $\V=\{V_s:\X\rightarrow\cS(\cE):s\in\Sigma\}$ be compound cq-channels.  
We call the pair $(\W,\V)$ a \emph{compound wiretap cqq-channel}. 
The channel output of $\W$ is available to the legitimate receiver (Bob) and the 
channel output of $\V$ is available to the wiretapper (Eve). 
We may sometimes write the channel as a family of pairs 
$(\W,\V)=\left(W_t,V_s\right)_{t \in \Theta, s\in\Sigma}$. 
\end{Definition}

\begin{Definition}
An $(n,M,\lambda)$ transmission code for the compound wiretap cqq-channel
$(W_t,V_s)_{t \in \Theta, s\in\Sigma}$ consists of a family $\C=(P_i,D_i)_{i\in[M]}$ 
where the $P_i$ are probability distributions on
$\cX^n$ and $(D_i)_{i\in[M]}$ 
a POVM on ${\cB}^{\otimes n}$ such that  
\begin{align*}
  \forall i\in[M]  \quad & \sup_{t\in\Theta} 1-\tr W^{\otimes n}(P_i)\!\cdot\!D_i \leq \lambda, \\
  \forall i,j\in[M]\quad & \sup_{s\in\Sigma} \frac12 \| V_s^{\otimes n}(P_i) - V_s^{\otimes n}(P_j) \|_1 \leq \mu.
\end{align*}
\end{Definition}
The capacity is defined as before.

\begin{Theorem}[\cite{BCDN17}]
The secrecy capacity of a compound wiretap cqq-channel $(\W,\V)$ is given by
\begin{align*}
  C_{S}(\W,\V) &= \lim_{n\rightarrow \infty} 
                    \sup_{U \rightarrow X^n \rightarrow (B_t^nE_s^n)} \\
               &\phantom{===}
                   \frac1n \left(\inf_{t\in\Theta} I(U;B^n_t)
                                    - \sup_{s\in\Sigma} I(U;E^n_s) \right),
\end{align*}
where $B_t$ are the resulting random quantum states at the
output of legal receiver channels and $E_s$
are the resulting random quantum states at the output of wiretap channel.
\qed
\end{Theorem}

\begin{Definition}
An \emph{$(n,N,\lambda_1,\lambda_2,\mu)$ wiretap ID-code} for the 
compound wiretap cqq-channel $(\W,\V)$ is a set of pairs 
$\{(P_i,D_i):i\in[N]\}$ where the $P_i$ are probability distributions on
$\X^n$ and the $D_i$, $0 \le D_i \leq1$, denote operators on 
$\cB^{\otimes n}$ such that for all $i \neq j\in[N]$ and all $0\leq F\leq \1$,
\begin{align}
  \inf_{t\in\Theta} \tr W_t^{\otimes n}(Q_i) D_i &\geq 1-\lambda_1, \nonumber \\
  \sup_{t\in\Theta} \tr W_t^{\otimes n}(Q_j) D_i &\leq \lambda_2,   \nonumber \\
  \inf_{s\in\Sigma}  \tr V_s^{\otimes n}(Q_j) F + \tr V_s^{\otimes n}(Q_i) (\1-F) &\geq 1-\mu.   \label{CQCCIDe} 
\end{align}
We define $N(n,\lambda_1,\lambda_2,\mu)$ as the largest $N$
satisfying the above definition for a given $n$ and set $\lambda_1,\lambda_2,\mu$
of errors.
\end{Definition}

\begin{Definition}
The secure identification capacity $C_{SID}(\W,\V)$ of a compound
wiretap cqq-channel $(\W,\V)$ is defined as $C_{SID}(\W,\V) \=$
\[
   \inf_{\lambda_1,\lambda_2,\mu>0} 
                     \liminf_{n\rightarrow\infty} \frac1n \log\log N(n,\lambda_1,\lambda_2,\mu).
\]
\end{Definition}

\medskip\noindent
Again we get a dichotomy result.

\begin{Theorem}
\label{QCWCdich}
Let $(\W,\V)$ be a compound wiretap cqq-channel. Then,
\begin{align*}
  C_{\text SID}(\W,\V)
   &= C^{\rm sim}_{\text SID}(\W,\V) \\
   &= \begin{cases}
              C(\W) & \text{if } C_S(\W,\V) > 0,\\
              0    & \text{if } C_S(\W,\V) = 0.
       \end{cases}
\end{align*}
\end{Theorem}

\begin{proof}
For the direct part, the identification code is again constructed by means of 
two fundamental codes, following Ahlswede and Dueck \cite{AhlswedeDueck:ID.B}.

Let $0<\epsilon<C$ be fixed. 
We know that there is a 
$\delta>0$ such that for sufficiently large $n$ there is an
$(n,M',\lambda(n))$-code for the compound cq-channel
\begin{equation}
  \label{qorgcode1}
  C'=\left\{\left(u'_j,\D'_j | j\in [M']\right)\right\}
\end{equation}
and an $(\lceil\sqrt{n}\,\rceil,M'',\lambda(\sqrt{n}))$-code for the compound wiretap cqq-channel
\begin{equation}
  \label{qorgcode2}
  C''=\left\{\left(u''_k,\D''_k | k\in [M'']\right)\right\}
\end{equation}
with code size $M'=\lceil 2^{n(C(\W)-\epsilon)}\rceil$ and
$M''=\lceil2^{\epsilon\sqrt{n}} \rceil$.
Now consider a family of maps 
$\left(T_i|i\in [N]\right)$, 
$T_i:[M'] \to [M''],\ \forall\  i\in [N]$,
where $T_i(j)$ yields the colour of 
code word $i$ under colouring number $j$. Thus we could construct an ID-code for the compound 
wiretap cqq-channel
$\{(Q_i,\D_i)|i\in [N]\}$ in the following way: Let
\[
  Q_i(x^n)=\left\{
   \begin{array}{c@{\quad}l}
    \frac{1}{M'} & \text{if }\exists j:x^n=u'_j\cdot u''_{T_i(j)} \\
    0 & \text{otherwise}.
  \end{array}
  \right.
\]
This means that we choose a colouring at random and calculate the 
corresponding colour of our message. We define the POVMs as
\[
  \D_i = \sum_{j=1}^{M'}\D'_j\otimes \D''_{T_i(j)}.
\]
Now we will show by random choice of the family of maps that there exists
a family which induces an ID-code for the cq-channel with the desired error probabilities.
For $i\in [N]$ and $j\in [M']$, take independent random
variables $U_{ij}$ with uniform distribution on the set 
$\{u'_j \cdot u''_k|k\in [M'']\}$. Collecting all RVs for one
message $i$ we get the random colour sets 
$$
\bar{U_i}=\{U_{i1}, \cdots, U_{iM'}\} \quad i \in \N
$$
and we will use equidistribution on $\bar{U_i}$ (which is 
a random probability distribution) as encoding distribution
$\bar{Q}_i$ for
message $i$. Therefore
$$
\D(\bar{U_i})=\sum_{j=1}^{M'}\D(U_{ij})
$$
with
$$
\D(U_{ij})=D'_j \otimes \D''_k,\quad \mbox{if } U_{ij}=u'_j\cdot u''_k.
$$
The random ID-code for the compound cq-channel is
$$
 \left\{ \left( \bar{Q}_i, 
 \D(\bar{U}_i \right) | i\in\N\right\}.
$$

For errors of the first kind we have for all possible realisations $U_i$ of
$\bar{U_i}$
\begin{equation*}
\min_{t\in\Theta} \tr W_t^{\otimes n}(Q_i) D(U_i) \ge 1-(\lambda(n)+\lambda(\sqrt{n})).
\end{equation*}
Thus errors of the first kind
tend to zero for $n\to\infty$.

Now we need to prove that with positive probability we get a code
with sufficiently small probability for errors of the second kind. 
Then there is a realisation with this error probability, and therefore
we are done.

We will analyse the overlapping between the $U_i$ (which determines
the probability for errors of the second kind as we use equidistribution
on the $U_i$ as encoding distribution for message $i$). To do this, 
we will define a Bernoulli chain counting the intersecting elements between
a realisation of $\bar{U_1}$ and $U_2$:
Let $U_1$ be fixed and define, for $j\in [M']$,
\[
  \Psi_j = \Psi_j\left(\bar{U_2}\right)
         = \begin{cases}
             1 & \text{ if } U_{2j}\in U_1,\\
             0 & \text{ if } U_{2j}\not\in U_1.
           \end{cases}
\]

This means that $\Psi_j=1$ iff messages 1 and 2 get the same colour under
colouring $j$. The RVs $U_{2j}$ are independent, therefore the $\Psi_j$ are
independent with
$$
E\Psi_j=\frac{1}{M''}.
$$

For $M''=\left\lceil 2^{\sqrt{n}\epsilon}\right\rceil$, we have
\begin{align*}
  D&\left(\lambda \bigg\|\frac{1}{M''}\right)\\
  &= \lambda\log\left(\lambda \left\lceil 2^{\sqrt{n}\epsilon}\right\rceil\right)+
 (1-\lambda)\log\frac{1-\lambda}
   {1-\left\lceil 2^{\sqrt{n}\epsilon}\right\rceil} \\
  &= \lambda\log\left\lceil2^{\sqrt{n}\epsilon}\right\rceil+
 \lambda\log\lambda+
 (1-\lambda)\log(1-\lambda)\\
  & 
- \; (1-\lambda)
  \log \left(1-\frac{1}{\left\lceil 2^{\sqrt{n}\epsilon}\right\rceil}\right)\\
  &\geq 
 \lambda\log\left(2^{\sqrt{n}\epsilon}\right)-
 \log\left(1-\frac{1}{2}\right)\\
  &= \lambda\cdot\sqrt{n}\cdot\epsilon+1.
\end{align*}

Now consider a realisation $U_2$ of $\bar{U_2}$. We have
\begin{equation*}
  \forall u\in U_1\setminus U_2: \; \max_{t\in\Theta} \tr(W_{t,u}^n\D(U_2))\le 2^{-n\delta} +
  2^{-\sqrt{n}\delta}.
\end{equation*}
This follows immediately from the error bounds of our original transmission
codes~(\ref{qorgcode1}) and~(\ref{qorgcode2}).

If now $\lambda\in\left(0,1\right)$ is given, we get that with positive 
probability, the events
\begin{equation}
  \label{qeq18}
  \max_{t\in\Theta} \tr W_{t,u}^nD\left(\bar{U_2}\right) \le
  \lambda+2\cdot 2^{-2\sqrt{n}\delta}
\end{equation}
and
\begin{equation}
  \label{qeq19}
  \max_{t\in\Theta}\tr W_{t,u}^nD\left(U_1\right) \le
  \lambda+2\cdot 2^{-2\sqrt{n}\delta}
\end{equation}
occur, provided that $n$ is large enough. Therefore there is a realisation
$U_2$ of $\bar{U_2}$ for which inequalities~(\ref{qeq18}) 
and~(\ref{qeq19}) hold, which leads to a code of size two. Repeating 
this argument for $i=3,\ldots,N$ and upperbounding the probability that 
the newly selected $U_i$ does not fulfil inequalities analogous to~(\ref{qeq18})
and~(\ref{qeq19}) for a certain $U_j, j\in [N-1]$ instead of 
$U_1$, by the sum of the probabilities for each $U_j$, we get that an 
$\left(n,N,\lambda+2\cdot2^{-\sqrt{n}\delta}\right)$-code exists, if
\begin{equation*}
\left(N-1\right) \Pr\left\{ \sum_{j=1}^{M'}\Psi_j>M'\lambda\right\} < 1.
\end{equation*}
But if $N\le 2^{2^{n(C-\epsilon)}(\lambda\sqrt{n}\epsilon-1)}$, 
then $N-1 < 2^{M' (\lambda\sqrt{n}\epsilon-1)}$. 
Therefore $2^{-M' (\lambda\sqrt{n}\epsilon-1)} < \frac{1}{N-1}$, and hence
by Hoeffding's bound
\[
  \Pr\left\{ \sum_{j=1}^{M'}\Psi_j>M'\lambda\right\} < \frac{1}{N-1}.
\]
Thus for all $\lambda\in (0,1)$ and for all $\epsilon>0$ 
\begin{align*}
  \lim_{n\to\infty}\frac{1}{n}& \log\log N\left(n,\lambda\right) \\
  &\geq \lim_{n\to\infty}
          \frac{1}{n} \log\log 2^{2^{n(C-\epsilon)}(\lambda\sqrt{n}\epsilon-1)}        \\
  &\geq \lim_{n\to\infty} \frac{n(C-\epsilon) + \log (\lambda\sqrt{n}\epsilon-1)}{n} \\
  &=    C-\epsilon.
\end{align*}
As in the classical case, 
it follows from the construction of the code for the compound wiretap cqq-channel 
that the wiretapper can not identify the second part of the message, and therefore 
condition (\ref{CQCCIDe}) is satisfied.

It is clear that the capacity $C_{SID}(\W,\V)$ of the channel with a wiretapper 
cannot be bigger than the capacity $C_{ID}(\W)$ of the channel without a wiretapper. 
Therefore it remains to be shown that 
$C_{SID}(\W,\V)=0$ necessarily, if $C_{S}(\W,\V)=0$.
We will show the contrapositive, that if $C_{SID}(\W,\V)>0$, then $C_{S}(\W,\V)> 0$.
Recall Lemmas~\ref{Lemma 3} and~\ref{Lemma 4} from Section~\ref{Secure},
from which we get directly (denoting by $Q$ the uniform distribution on
a set $\{i,j\}$ of two messages):

\begin{itemize}
\item If for $i\neq j\in [N]$, it holds for all POVM elements 
$F$ on ${\cB}^{\otimes n}$ and all $s\in\Sigma$ that
\[
  \tr(V_s^{\otimes n}(P_j) F) + \tr(V_s^{\otimes n}(P_i) (\1-F)) \geq 1-\lambda,
\]
then
\[
  \sup_{s\in\Sigma} I(Q;V_s^{\otimes n}) \leq h\left(\frac{\lambda}{2}\right).
\]

\item If for $i\neq j\in [N]$, it holds for all $t\in\Theta$ that
\begin{align*}
   \tr(W_t^{\otimes n}(P_i) D_i) &\ge 1-\lambda, \text{ and} \\
   \tr(W_t^{\otimes n}(P_j) D_i) &\le \lambda,
\end{align*}
then
\[
  \min_{t\in\Theta} I(Q;W_t^{\otimes n}) \geq 1-h(\lambda).
\]
\end{itemize}

The existence of an ID-wiretap code with positive rate implies the
above conditions. Hence, as before, we obtain if 
$\lambda \leq \frac{1}{15}$, then 
\begin{align*}
  C_{S}(\W,\V) &\geq \frac1n \left( \inf_{t\in\Theta} I(Q;W_t^{\otimes n}) 
                                   - \sup_{s\in\Sigma} I(Q;V_s^{\otimes n}) \right)\\ 
               & > 0,
\end{align*}
and we are done.
\end{proof}

\section{Identification in the \protect\\ presence of a jammer}
\label{Jammer}
In this section we perform the same analysis for the case of arbitrarily varying 
cq-channels, with analogous findings. We point out, however, that we
only consider \emph{finite} index sets $\Theta$ throughout this and the
following section.

\begin{Definition}
\label{CQAVC}
Let $\Theta$ be a finite index set, 
$\X$ a finite set and $\B$ a finite-dimensional Hilbert space. 
Let $W_t:\X \longrightarrow \cS(\B)$ be a cq-channel for every $t\in \Theta$:
\[
  W_t:\X\ni x\mapsto W_{t}(x)\in \mathcal{S}(\B), \quad t\in \Theta.
\]
Let $t^n\in\Theta^n$ be a state sequence.
The memoryless extension of the cq-channel $W_{t^n}$ is given by 
$W_{t^n}(x^n) =  W_{t_1}(x_1)\otimes \ldots \otimes W_{t_n}(x_n)$ for $x^n \in \X^n$.
We call $\W \= \{W_t\}_{t\in \Theta}$ \emph{an arbitrarily varying cq-channel}.
\end{Definition}

In this case a jammer can change the channel during the transmission.
\begin{Definition}
An \emph{$(n,M,\lambda)$-code} for the arbitrarily varying cq-channel $\W$ is a family 
$\mathcal{C}\=\left( (P_m,D_m) : m\in[M] \right)$ consisting of pairs of stochastic
encodings given by code word probability distributions $P_m$ over $\X^n$ 
and positive semi-definite operators  $D_i$ on $\B^{\otimes n}$, 
forming a POVM, i.e.~$\sum_{m=1}^{M} D_m = \1$, such that
\begin{equation*}
   \max_{t^n\in\Theta^n} \max_{i\in[M]} 1-\tr W_{t^n}^{\otimes n}(P_i) D_i  
    \leq \lambda.
\end{equation*}
Like in the compound case, here we allow explicitly stochastic encoders. 
The number $M$ is called the \emph{size} of the code, and $\lambda$ the error 
probability. The maximum $M$ for given $n$ and $\lambda$ is
denoted $M(n,\lambda)$, extending the definition for a cq-channel
(which is recovered for $|\Theta|=1$).

The capacity of $\W$ is defined as before,
\[
  C(\W) = \inf_{\lambda>0} \liminf_{n\rightarrow\infty} \frac1n \log M(n,\lambda).
\]
\end{Definition}

A more intuitive description of the arbitrarily varying cq-channel is that 
a jammer tries to prevent the
legal parties from communicating properly. He may change his input in every 
channel use and is not restricted 
to use a repetitive probabilistic
strategy. Quite on the contrary, it is understood that the sender and the receiver
have to select their coding scheme first. After that the jammer makes his choice
of the sequence of channel states.
The sender and receiver do not know which channel from the set $\W$ 
is actually used; their prior knowledge is 
merely that the channel is memoryless and belongs to the set $\W$. Their task
is to prepare for the worst case among those.


\begin{Definition}
We say that the arbitrarily varying cq-channel
$\W = \{W_t : t \in \Theta\}$ is \emph{symmetrizable} if there exists a
parametrized set of distributions 
$\{\tau(\cdot|x): x\in\X\}$, on $\Theta$
also known as a channel $\tau$ from $\X$ to $\Theta$,
such that for all $x,x'\in\X$,
\[
  \sum_{t\in\Theta}\tau(t|x)W_{t}(x') = \sum_{t\in\Theta}\tau(t|x')W_{t}(x).
\]
\end{Definition}

To formulate the capacity theorem of \cite{AB07},
we need the following notations.
For an arbitrarily varying cq-channel $\mathcal{W}$ we denote its convex hull by $\conv(\mathcal{W})$. 
It is defined as follows: 
\begin{align*}
 \conv(\mathcal{W})
     =\left\{W_{q} : W_q=\sum_{t\in \Theta}q(t)W_t,\ q \in\P(\Theta), \right\}.
\end{align*}

Furthermore, we set
\begin{align*}
  C_{\text{ran}}(\W) \= \max_{p\in\P(\X)} \min_{W\in\conv(\W)} I(p;W). 
  \label{avcqc_rand_dir}
\end{align*}
This is called the random coding capacity of the channel. 
Under this notion, the encoding with a stochastic encoder is 
generalized to a (correlated) random code.
It is assumed that the sender and the receiver have access to some 
source with correlated randomness, which, however, is secret from the 
jammer. Here we need just the quantity 
to give the transmission capacity of the arbitrarily varying cq-channel.

\begin{Theorem}[\cite{AB07}]
Let $\W$ be an arbitrarily varying cq-channel. Then its capacity
$C(\W)$ is given by
\[
  \hspace{0.5cm}
  C(\W) = \begin{cases}
            0                  & \text{if $\W$ is symmetrizable,} \\
            C_{\text{ran}}(\W) & \text{otherwise.}
                                 \hspace{2.3cm} \blacksquare
          \end{cases}
\]
\end{Theorem}

\begin{Definition}
An \emph{$(n,N,\lambda_1,\lambda_2)$ ID-code} for the arbitrarily varying cq-channel
$\W$ is a set of pairs $\{(P_i,D_i):i\in[N]\}$, where the $P_i$ are 
probability distributions on $\X^n$ and the $D_i$ are POVM elements, 
i.e.~$0 \le D_i\le \1$, acting on $\B^{\otimes n}$, 
such that $\forall\, i\neq j \in [N]$
\begin{align*}
  \min_{t^n\in\Theta^n} \tr W_{t^n}(P_i)\!\cdot\!D_i &\geq 1-\lambda_1 \text{ and }\\
  \max_{t^n\in\Theta^n} \tr W_{t^n}(P_i)\!\cdot\!D_j &\leq \lambda_2.
\end{align*}
The largest size of an $(n,N,\lambda_1,\lambda_2)$ ID-code
is denoted $N(n,\lambda_1,\lambda_2)$. Analogous to previous
definitions, we have also simultaneous ID-codes and the
maximum code size $N_{\rm sim}(n,\lambda_1,\lambda_2)$. 
\end{Definition}

The identification capacities are now defined as before.

With the help of the method from Theorem~\ref{IDcompound} we can show that the transmission 
capacity of the channel corresponds to the identification capacity. 
To do this, in the proof of the direct part we simply use a code for an arbitrary varying cq-channel 
instead of the code for the compound cq-channel. 
To show the converse, we show that the error of the first type in the identification can not be arbitrarily small
if the channel is symmetrizable. Therefore, we get the following.

\begin{Theorem}
Let $\W$ be an arbitrarily varying cq-channel. Then its ID-capacity
is given by
\begin{align*}
  \hspace{0.4cm}
  C^{\rm sim}_{\rm ID}(\W)
    &= C_{\rm ID}(\W) \\          
    &= \begin{cases}
         0                  & \text{if $\W$ is symmetrizable,} \\
         C_{\text{ran}}(\W) & \text{otherwise.}
                              \hspace{2cm} \blacksquare
       \end{cases}
\end{align*}
\end{Theorem}

\section{Secure identification \protect\\ in the presence of a jammer}
\label{Securejammer}
In this section we add a wiretapper to the arbitrarily varying 
cq-channel. First we define the transmission codes and quote the 
known transmission capacity.
Using this result, we then calculate the secure identification 
capacity of the arbitrarily varying wiretap cqq-channel.

\begin{Definition}
Let $\Theta$ and $\Sigma$ be finite index sets, 
and let $\W=\{W_t:\X\rightarrow\cS(\cB):t\in\Theta\}$ 
and $\V=\{V_s:\X\rightarrow\cS(\cE):s\in\Sigma\}$ be arbitrarily varying cq-channels.  
We call the pair $(\W,\V)$ an \emph{arbitrarily varying wiretap cqq-channel}. 
The channel output of $\W$ is available to the legitimate receiver (Bob) and the 
channel output of $\V$ is available to the wiretapper (Eve). 
We may sometimes write the channel as a family of pairs 
$(\W,\V)=(W_t,V_s)_{t\in\Theta,s\in\Sigma}$. 
\end{Definition}

\begin{Definition}
An $(n,M,\lambda)$ transmission code for the arbitrarily varying wiretap cqq-channel
$(W_t,V_s)_{t \in \Theta, s\in\Sigma}$ consists of a family $\C=(P_i,D_i)_{i\in[M]}$,
where the $P_i$ are probability distributions on
$\cX^n$ and $(D_i)_{i\in[M]}$ 
a POVM on ${\cB}^{\otimes n}$ such that  
\begin{align*}
  \forall i\in[M]  \quad & \max_{t^n\in\Theta^n} 1-\tr W_{t^n}^{\otimes n}(P_i)\!\cdot\!D_i \leq \lambda, \\
  \forall i,j\in[M]\quad & \max_{s^n\in\Sigma^n} \frac12 \| V_{s^n}^{\otimes n}(P_i) - V_{s^n}^{\otimes n}(P_j) \|_1 \leq \mu.
\end{align*}
\end{Definition}

The capacity is defined as before.
To state the result of \cite{BCDN18} we again introduce
the random coding capacity,
\begin{align*}
  C_{\text{S,ran}}&(\W,\V) \= \lim_{n\rightarrow \infty} \frac1n 
                              \max_{U\rightarrow X^n \rightarrow B_{t^n}^n E^n_{s^n}} \\
      &\phantom{=======}
       \Biggl( \min_{\widehat{W} \in \operatorname{conv}\{W_{t}:t\in \Theta\}} 
                                    I(p_U;\widehat{W}^{\otimes n})  \\
      &\phantom{=============}
	           - \max_{s^n\in \Sigma^n} I(p_U;V_{s^n}) \Biggr).
\end{align*}
Here, $B_{t^n}^n$ are the resulting quantum states at the output of the
legitimate receiver's channels. $E^n_{s^n}$ are the resulting  quantum states  at
the output of the wiretap channels.
The maximum is taken over all random
variables  that satisfy the Markov chain relationships:
$U\rightarrow X^n \rightarrow B_{t^n}^n E^n_{s^n}$.  
$X^n$ is here a random variable taking values in $\X^n$, $U$ a random
variable taking values on some finite set $\cU$
with probability  distribution $p_U$.
In \cite{BCDN18} the following dichotomy is shown.

\begin{Theorem}[\cite{BCDN18}]
\label{multi1}
Let $C_S(\W,\V)$ denote the capacity of the arbitrarily varying wiretap 
cqq-channel $(\W,\V)$. Then,
\[
   C_S(\W,\V) 
          = \begin{cases}
              0                       & \text{if $\W$ is symmetrizable,} \\
              C_{\text{S,ran}}(\W,\V) & \text{otherwise}.
                                        \hspace{1.6cm} \blacksquare
            \end{cases}
\]
\end{Theorem}

As in the previous section, we can now use a similar proof
technique to determine the secure identification capacity of the 
arbitrarily varying wiretap cqq-channel. We start by defining the identification codes.

\begin{Definition}
An \emph{$(n,N,\lambda_1,\lambda_2,\mu)$ wiretap ID-code} for the 
arbitrarily varying wiretap cqq-channel $(\W,\V)$ is a set of pairs 
$\{(P_i,D_i):i\in[N]\}$ where the $P_i$ are probability distributions on
$\X^n$ and the $D_i$, $0 \le D_i \leq1$, denote operators on 
$B^{\otimes n}$ such that $\forall i,j\in[N],\ i\neq j$ and $0\leq F\leq \1$,
\begin{align*}
  \min_{t^n\in\Theta^n} \tr W_{t^n}^{\otimes n}(Q_i) D_i &\geq 1-\lambda_1, \\
  \max_{t^n\in\Theta^n} \tr W_{t^n}^{\otimes n}(Q_j) D_i &\leq \lambda_2, \\
  \label{CQCCIDer} 
  \min_{s^n\in\Sigma^n} \tr V_{s^n}^{\otimes n}(Q_j) F + \tr V_{s^n}^{\otimes n}(Q_i) (\1-F) &\geq 1-\mu.
\end{align*}
We define $N(n,\lambda_1,\lambda_2,\mu)$ as the largest $N$
satisfying the above definition for a given $n$ and set $\lambda_1,\lambda_2,\mu$
of errors.
\end{Definition}

\begin{Definition}
The identification capacity $C_{SID}(\W,\V)$ of an arbitrarily varying
wiretap cqq-channel $(\W,\V)$ is defined as 
\[\begin{split}
  C_{SID}&(\W,\V) \\
         &\= \inf_{\lambda_1,\lambda_2,\mu>0}
                     \liminf_{n\rightarrow\infty} \frac1n \log\log N(n,\lambda_1,\lambda_2,\mu).
\end{split}\]
\end{Definition}

Again we show a dichotomy result, using the idea of Theorem~\ref{QCWCdich}. 
As fundamental codes we use for $C'$ a code for the arbitrarily varying cq-channel and
for $C''$ a code for the arbitrarily varying wiretap cqq-channel, both 
reaching the capacity. If the transmission capacity for $C''$ is positive,
we get as an identification capacity the transmission capacity of $C''$. 
The security follows by the strong secrecy condition like in 
Theorem~\ref{QCWCdich}. Also the converse follows the same idea.
Therefore we get the following.
 
\begin{Theorem}[Dichotomy]\label{AVCWCdich}
Let $C_{}(\W)$ be the capacity of the arbitrarily varying cq-channel $\W$ and let $C_{S}(\W,\V)$ be
the secrecy capacity of the arbitrarily varying wiretap cq-channel $(\W,\V)$. Then,
\begin{align*}
  \hspace{0.5cm}
  C_{\text SID}(\W,\V)
     &= C^{\rm sim}_{\text SID}(\W,\V) \\
     &= \begin{cases}
              C(\W) & \text{if } C_S(\W,\V) > 0,\\
              0     & \text{if } C_S(\W,\V) = 0.
                         \hspace{0.6cm}\blacksquare
        \end{cases}
\end{align*}
\end{Theorem}

In this theorem the capacity is a single letter formula, but the condition 
if the capacity is positive is given by the multi-letter formula
for the random coding secret capacity.

\begin{rem}
In the case of transmission it is possible to avoid the capacity 
being zero if the channel is symmetrizable, if we allow the 
sender and receiver to use common randomness. With this resource the capacity 
will not change if the channel is non-symmetrizable, but if the channel is symmetrizable
then the capacity may go up from zero to the random coding capacity. 

The situation appears different in the case of identification. We can, of course,
use the same resource to get rid of the vanishing capacity in the symmetrizable case. 
However, note that a positive rate of common randomness, by the
concatenated code construction of Ahlswede and Dueck \cite{AhlswedeDueck:ID.B},
increases the ID-capacity by the same amount. 
Fortunately, it comes to our rescue the fact that whenever common randomness is
required to achieve the random coding capacity for transmission, then
a rate of asymptotically zero is sufficient \cite{Ahlswede:elim}. Thus, we
could define random coding capacities with zero rate of common randomness
without changing the notion for transmission, while obtaining a sound
capacity concept for the identification problem. 
\end{rem}

\section{Continuity and super-additivity}
\label{Continuity}
In \cite{BD17} we discussed the continuity and super-additivity for the identification capacity
of a classical compound channel and a classical compound wiretap channel. It turns 
out that the results for the capacity of the classical-quantum case are completely
analogous. Therefore we just list the results here and discuss them as briefly
as possible.

%

\subsection{Distance between cq-channels}
First we need a metric to measure the distance between two cq-channels.

\begin{Definition}
Let $W,\widetilde W:{\X}\longrightarrow {\cS}({\cB})$ be two cq-channels.  
The distance between them is defined as
\begin{equation*}
  d(W,\widetilde W) \= \max_{x\in\X} \left\lVert W(x)-\widetilde W(x)\right\rVert_1,
\end{equation*}
where $\lVert \cdot \rVert_1$ denotes the trace norm.
\end{Definition}

Next, we extend this concept to the compound and arbitrarily varying case.
\begin{Definition}
Let $\W=(W_t)_{t\in\Theta}$ and $\widetilde\W=(\widetilde W_s)_{s\in\widetilde\Theta}$ be two compound or arbitrarily varying cq-channels with input alphabet $\X$ and let
\begin{equation*}
  G(\W,\widetilde\W) \= \sup_{t\in \Theta} \inf_{t'\in\widetilde\Theta} d(W_{t},\widetilde{W}_{t'}).
\end{equation*}
Then we define the distance between the two cq-channels as
\begin{equation*}
  D(\W,\widetilde\W) \= \max\{G(\W,\widetilde\W),G(\widetilde\W,\W)\}.
\end{equation*}
\end{Definition}

Obviously, it is desirable to have a \emph{continuous} behaviour of the capacity,
meaning that small variations in the channel (i.e.~the set $\W$) 
set should only lead to small variations in the corresponding capacity. 
Let ${\mathfrak W}(\X,{\cS}({\cB}))$ be the family of all compound cq-channels 
$\W=(W_t)_{t\in\Theta}$ with $W_t:{\X}\longrightarrow {\cS}({\cB})$,
with respect to the above metric $D$.
We use the distance definition to define continuity for points and functions in the usual way.

%
Similarly, for wiretap cqq-channels $(\W,\V)$,
the metric to measure the distance between two 
wiretap cqq-channels is as follows.

\begin{Definition}
\label{measure}
Let $(W,V)$ and $(\widetilde W,\widetilde V)$ be two
wiretap cqq-channels with the same input alphabet $\X$, then we define
\begin{equation*}
  d_S((W,V),(\widetilde W,\widetilde V)) \= \max \{d(W,\widetilde W),d(V,\widetilde V)\}.
\end{equation*}
If $(\W,\V)$ and $(\widetilde \W, \widetilde\V)$ are two compound/arbitrarily varying
wiretap cqq-channels with the same input alphabet $\X$. Then we define
\begin{equation*}
  D_S\bigl( (\W,\V),(\widetilde \W, \widetilde\V) \bigr) 
             \= \max \{D(\W,\widetilde \W), D(\V,\widetilde\V)\}.
\end{equation*}
\end{Definition}
The notions of the (dis-)continuity points are as usual.


We also consider parallel (i.e. tensor product) channels, which means that
they map pair of inputs independently to a tensor product of the output systems:
define $\W\otimes \widetilde\W$ as the set of channels
\[
  W_{t_1}\otimes \widetilde{W}_{t_2}: \X_1\times\X_2 \to {\cS}({\cB_1})\otimes {\cS}({\cB_2}),
\]
with
\[
  W_{t_1}\otimes \widetilde W_{t_2}(x_1,x_2)\=W_{t_1}(x_1)\otimes \widetilde W_{t_2}(x_2).
\]

Let $\W$ be a compound cq-channel.
From Theorem~\ref{cqid}, we know that
\[
  C_{ID}(\W)= C(\W) =\max_P \min_{t\in\Theta} I(P;W_t). 
\]
This is a continuous function of $\W$ and therefore the following holds.

\begin{Corollary}
  $C_{ID}$ is a continuous function on ${\mathfrak W}(\X,{\cS}({\cB}))$.
  \qed
\end{Corollary}

Regarding the additivity, we can once more use Theorem~\ref{cqid}. 
It follows immediately that

\begin{Corollary}
For any two compound cq-channels $\W$ and $\widetilde\W$,
\begin{equation*}\hspace{1cm}
  C_{ID}(\W\otimes\widetilde\W) = C_{ID}(\W)+C_{ID}(\widetilde\W).
  \hspace{1cm}\blacksquare
\end{equation*}
\end{Corollary}

\begin{Definition}
We say that a capacity $C$ is super-additive if we can find two channels $(W,V)$ 
and $(\widetilde W,\widetilde V)$ such that 
\begin{equation}
  \label{super}
  C(W\otimes \widetilde W, V\otimes \widetilde V) > C(W,V) + C(\widetilde W, \widetilde V).
\end{equation}
\end{Definition}

The following theorem characterizes the discontinuity points of $C_{SID}$ 
completely. It also shows that the set of discontinuity points is never empty.

\begin{Theorem}
\label{disSID}
The wiretap cqq-channel $(W,V)$ is a discontinuity point of $C_{SID}$ iff
\begin{enumerate}
\item $C(W)>0$,
\item $C_S(W,V)=0$, and
\item For each $\epsilon>0$ there exists a wiretap channel $(W_\epsilon,V_\epsilon)$ such that
      $d_S((W,V),(W_\epsilon,V_\epsilon)) < \epsilon$ and
      $C_S(W_\epsilon,V_\epsilon) > 0$.

\end{enumerate}
\end{Theorem}

The proof follows the same idea as the proof for the classical case in \cite[Thm.~6.1]{BD17}.
For each cq-channel $W$ with $C(W)>0$, there exists a channel $V_*$, such that 
for $(W,V_*)$ the conditions 2 and 3 of Theorem~\ref{disSID} are fulfilled.
Therefore $(W,V_*)$ is an example of a discontinuity point of $C_{SID}$.
Because of that, there is a huge number of discontinuity points.

\begin{Corollary}
Let $(W,V)$ be a wiretap cqq-channel with $C_{SID}(W,V)>0$. 
Then there exists a $\hat{\epsilon}>0$, 
such that for all $(\widetilde W,\widetilde V)$ with 
$D((W,V),(\widetilde W,\widetilde V))<\hat{\epsilon}$, it holds that
$C_{SID}(\widetilde W,\widetilde V)>0$.
\qed
\end{Corollary}

The identification capacity of a compound cq-channel is additive and therefore 
not super-additive. 
It follows by its operational definition that for the message transmission 
capacity and for the message transmission secrecy 
capacity, inequality (\ref{super}) holds with ``$\geq$''. 

By the same argument, we can show that the same also holds 
for the secure identification capacity.
\begin{Proposition}
For any two wiretap cqq-channels $(W,V)$ and $(\widetilde W,\widetilde V)$,
\begin{equation*}
  C_{SID}(W\otimes\widetilde W,V\otimes\widetilde V)
          \geq C_{SID}(W,V)+C_{SID}(\widetilde W,\widetilde V).
\end{equation*}
\end{Proposition}
\begin{proof}
  This follows from the coding theorem.
\end{proof}

\medskip
The following theorem gives a complete characterization of the super-additivity behaviour of 
$C_{SID}$.

\begin{Theorem}
\label{supI}
Let $(W,V)$ and $(\widetilde W, \widetilde V)$ be two wiretap cqq-channels.

1) Assume $\min\{C(W),C(\widetilde W)\}>0$. Then,
        $$
          C_{SID}(W\otimes\widetilde W, V\otimes\widetilde V)
            > C_{SID}(W, V)+C_{SID}(\widetilde W,\widetilde V)
        $$
        holds iff $C_S(W\otimes\widetilde W, V\otimes\widetilde V) > 0$,
        but at least one of $C_S(W,V)$ or $C_S(\widetilde{W},\widetilde{V})$ equals $0$.

2) Assume $C(W)=0$ [$C(\widetilde W)=0$]. Then, 
        $$
          C_{SID}(W\otimes\widetilde W, V\otimes\widetilde V)
            > C_{SID}(W, V)+C_{SID}(\widetilde W,\widetilde V)
        $$
        holds iff $C_S(W\otimes\widetilde W, V\otimes\widetilde V) >0$, but
        $C_S(\widetilde W, \widetilde V)=0$ [$C_S(W, V)=0$].
\end{Theorem}
The \textbf{proof} follows the same idea as the proof for the classical case 
\cite[Thm.~6.2]{BD17}.
\qed

\medskip
\begin{Theorem}
Let $(\W,\V)$ be a compound wiretap cqq-channel.
$(\W,\V)$ is a discontinuity point of $C_{SID}$ if the following properties are fulfilled:
\begin{enumerate}
\item $C(\W)>0$
\item $C_S(\W,\V)=0$
\item For all $\epsilon >0$ there exists a CWC $(\W_\epsilon, \V_\epsilon)$ with 
\(
D_S((\W,\V),(\W_\epsilon, \V_\epsilon))<\epsilon\ {\rm and}\ C_S(\W_\epsilon, \V_\epsilon)>0.
\)
\end{enumerate}
\end{Theorem}
The \textbf{proof} follows the same idea as the proof for the classical 
case \cite[Thm.~6.3]{BD17}.
\qed

\medskip
As before, we find a large number of discontinuity points. 
Now we characterize the super-additivity of these channels.
Theorem~\ref{supI} can be generalized for compound wiretap cqq-channels.
Furthermore, we consider the sharpest form of super-additivity, that is, super-activation.

\begin{Definition}
We say that a capacity $C$ can be super-activated if we can find 
two cqq-channels $(\W,\V)$ and 
$(\widetilde \W,\widetilde \V)$ such that 
\begin{equation*}
  C(\W\otimes \widetilde \W, \V\otimes \widetilde \V)>0
    \ \text{ and }\ 
  C(\W,\V) = C(\widetilde \W, \widetilde \V)=0.
\end{equation*}
\end{Definition}

\begin{Theorem}
Let $(\W,\V)$ and $(\widetilde\W, \widetilde\V)$ be two compound wiretap cqq-channels.
Then for these two channels
we have super-activation for $C_{SID}$ iff we have super-activation for $C_S$.
\end{Theorem}
The \textbf{proof} follows the same idea as the proof for the classical 
\cite[Thm.~6.4]{BD17}.
\qed

\medskip
The analysis of the transmission capacities of the arbitarily varying cq-channels and 
arbitrarily varying wiretap cqq-channels has been done in \cite{BCDN16} and \cite{BCDN17}.
There we showed that the transmission random coding capacity of the arbitrarily 
varying cq-channel is continuous. 

To give a complete characterization of the 
discontinuity points of the capacity as in \cite{BD17a}, let us introduce the set
\[
  \cN = \{ \W \text{ finite and symmetrizable} \}.
\]
Note that being symmetrizable is a closed condition, hence 
$\cN$ is a closed set under the convergence induced by the metric $D$.
%
%
%
With this, we can give a complete characterization of the 
discontinuity points of the capacity, just as in \cite{BD17a}.
\begin{Theorem} 
The capacity $C_{ID}(\W)$ is discontinuous at the finite cq-AVC $\W$ 
iff the following conditions hold:
\begin{enumerate}
  \item $C_{\text{ran}} (\W) > 0$
  \item $\W \in \cN$, i.e.~the channel is symmetrizable, 
        and for every $\epsilon > 0$ there exists a finite arbitrarily 
        varying cq-channel $\widetilde{\W}$ with $D(W,\widetilde{W}) < \epsilon$ and
        $\widetilde{\W} \not\in  \cN$.
        \qed
\end{enumerate}
\end{Theorem}
The second condition is precisely that $\W$ belongs to the boundary of $\cN$,
$
  \partial\cN = \bigl\{ \W \in \cN :\ \forall\epsilon > 0\, \exists \widetilde\W \text{ s.t. }
                        D(\widetilde\W,\W) < \epsilon \text{ and } \widetilde\W \not\in\cN \bigr\}.
$
The following result establishes a certain robustness property of the capacity.
It holds because $\cN$ is a closed set, hence every point outside it 
has a neighbourhood not intersecting it.

\begin{Theorem}
\label{theorem35} 
Let $\W$ be a finite arbitrarily varying cq-channel with $\W$ being not symmetrizable.
Then there exists an $\epsilon > 0$
such that all finite arbitrarily varying cq-channels $\widetilde{\W}$ with
$D(\widetilde{\W}, \W) < \epsilon$ 
are continuity points of $C_{ID}(\W)$.
\qed
\end{Theorem}

Let $\W$ be an arbitrary varying cq-channel and $\{\W_n\}^\infty_{n=1}$ 
be an arbitrary sequence of finite arbitrarily varying cq-channels with
\be\label{eq46}
\lim_{n\to\infty} D(\W_n , \W) = 0.
\ee
We define the variance of $C_{ID} (\W_n )$ for the sequence $\{W_n \} = \{W_n \}^\infty_{n=1}$ as
\[
  V(\{W_n \}) = \limsup_{n\to\infty} C_{ID} (\W_n ) - \liminf_{n\to\infty} C_{ID} (\W_n),
\]
and furthermore, let
\[
  \overline{V}(\W) = \sup V(\{\W_n \}),
\]
where the sup is taken over all $\{\W_n \}$ and $\W$ that 
satisfy (\ref{eq46}). In other words, $\overline{V}(\W)$ describes the 
maximal variation of $C_{ID}(\W)$ in the neighborhood of a certain channel $\W$.
Finally,
\[
  \overline{V} = \sup_{\W} \bar{V} (\W)
\]
is the maximal variation for all arbitrarily varying cq-channel $\W$.
Furthermore, we set $\cN_\infty\=\{\W : C_{\text{ran}}(\W)=0\}$.
Then we have the following result.

\begin{Theorem}
For a finite arbitrarily varying cq-channel $\W$, the following assertions hold:
\begin{enumerate}
  \item $\overline{V}(\W) = 0$ for $\W \not\in\partial\cN\backslash \cN_\infty$.
  \item $\overline{V}(\W) = C_{\text{ran}}(\W)$ for $\W \in \partial\cN \setminus \cN_\infty$.
  \item $\displaystyle\overline{V} = \sup_{\W\in \partial\cN \setminus \cN_\infty} C_{\text{ran}}(\W)$.
        \qed
\end{enumerate}
\end{Theorem}

\medskip
Now we will examine the additivity of the capacity function. 

\begin{Theorem}
Let $\W_1$ and $\W_2$ be two arbitrarily varying cq-channels. Then,
$C_{ID}(\W_1 \otimes \W_2 ) = 0$
iff 
$C_{ID} (\W_1 ) = C_{ID}(\W_2 ) = 0$.
\qed
\end{Theorem}

The next result shows that the ID capacity is super-additive.

\begin{Theorem}
Let $\W_1$ and $\W_2$ be two arbitrarily varying cq-channels. Then,
\[
C_{ID}(\W_1 \otimes \W_2 ) > C_{ID}(\W_1 ) + C_{ID} (\W_2 )
\]
iff exactly one of the two channels $\W_1$, $\W_2$ is symmetrizable
while the other one is not, and both random coding capacities are
positive, $C_{\text{ran}} (\W_1 ) > 0$, $C_{\text{ran}} (\W_2 ) > 0$.
\qed
\end{Theorem}

Now we will analyze the continuity of $C_{SID}$ for arbitrarily varying wiretap cqq-channels.
In Theorem~\ref{AVCWCdich} we showed that
$C_{SID}(\W,\V)=C_{SID}^{\rm sim}(\W,\V)=C(\W)$ if $C_S(\W,\V)>0$ and $=0$ otherwise.
We shall now use this result to fully characterize the continuity 
behavior and the discontinuity behaviour of $C_{SID}$.
To do so, we distinguish two cases,
1.~$C_{S,ran}(\W,\V)>0$ and 2.~$C_{S,ran}(\W,\V)=0$.

\begin{Theorem}
\label{C1} 
Let $(\W,\V)$ be a a finite arbitrarily varying wiretap cqq-channel
with $C_{\text{S,ran}}(\W,\V)>0$. Then,
$(\W,\V)$ is a discontinuity point of $C_{SID}$ iff
$\W$ is symmetrizable in the boundary of $\cN$, i.e.~$\W\in\partial\cN$.
\end{Theorem}
The \textbf{proof} follows the same lines as the argument in \cite{BD17a}.
\qed

\medskip
Furthermore, we have the following important stability results.

\begin{Theorem}
  Let $(\widehat{\W},\widehat{\V})$ be a cq-AVWC with 
  $C_{SID}(\widehat{\W},\widehat{\V})> 0$. Then there exists an $\epsilon > 0$ 
  such that for all finite cq-AVWCs $(\W,\V)$ with 
  $D((\W,\V),(\widehat{\W},\widehat{\V})) < \epsilon$, always $C_{SID}(\W,\V)> 0$. 
  In particular, $C_{SID}$ is continuous at $(\widehat{\W},\widehat{\V})$.
  \qed
\end{Theorem}

\begin{Theorem}
  \label{C2} 
  Let $(\W,\V)$ be a a finite arbitrarily varying wiretap cqq-channel
  with $C_{\text{S,ran}}(\W,\V)=0$. Then,
  $(\W,\V)$ is a point of discontinuity of $C_{SID}$
  iff $C_{\text{ran}}(\W)> 0$ and for every $\epsilon> 0$ there exists a 
  finite cq-AVWC $(\W_\epsilon,\V_\epsilon)$ with 
  $D((\W,\V),(\W_\epsilon,\V_\epsilon)) < \epsilon$, such that
  $\W_\epsilon$ is not symmetrizable
  and $C_{\text{S,ran}}(\W_\epsilon,\V_\epsilon)> 0$.
  \qed
\end{Theorem}


To end, we fully characterize the occurrence of super-activation and 
super-additivity for $C_{SID}$.
Of course, super-activation is the most powerful form of super-additivity,
in this case two channels each with capacity zero combine to one with positive capacity. 

It is known that $C_{ID}$ cannot be super-activated. For $C_{SID}$, 
a different behaviour can be observed:

\begin{Theorem}
Let $(\W_1,\V_1)$, $(\W_2,\V_2)$ be two arbitrarily varying wiretap cqq-channels. 
Then the following holds.
\begin{enumerate}
\item If $\max\{C_{\text{S,ran}}(\W_1,\V_1), C_{\text{S,ran}}(\W_2,\V_2)\}> 0$,
  then $C_{SID}$ shows super-activation for these two channels precisely when
  one of $C_{\text{S,ran}}(\W_1,\V_1)$ or $C_{\text{S,ran}}(\W_2,\V_2)$ equals $0$, and the
  other one is positive.
  W.l.o.g.~$C_{\text{S,ran}}(\W_2,\V_2) = 0$ and $C_{\text{S,ran}}(\W_1,\V_1) > 0$.
  Therefore, $\W_1$ necessarily is symmetrizable.
\item If $C_{\text{S,ran}}(\W_1,\V_1)=C_{\text{S,ran}}(\W_2,\V_2)=0$, 
  then $C_{SID}$ can be super-activated for these channels iff
  $\W_1 \otimes \W_2$ is not symmetrizable.
  (Note that this condition is equivalent to saying that at least
  one of the two channels $\W_1$ or $\W_2$ is not symmetrizable.)
  \qed
\end{enumerate}
\end{Theorem}

Next, we characterise the case in which we observe super-additivity, 
but in which no super-activation occurs.

\begin{Theorem}
Let $(\W_1,\V_1)$ and $(\W_2,\V_2)$ be two arbitrarily varying wiretap 
cqq-channels for which no super-activation occurs. 
Then, for these two channels, super-additivity of $C_{SID}$ 
applies iff 
$C_{SID}(\W_1,\V_1)> 0$ and $C_S(\W_2,\V_2) = 0$ but $C_{\text{ran}}(\W_2)> 0$,
or analogously with 1 and 2 interchanged.
\end{Theorem}

In \cite{BCDN16}, super-activation has been shown for the transmission 
capacity of the classical arbitrarily varying classical-quantum wiretap 
channels, and in \cite{BCDN17a}, a full characterization have been given.

\section{Conclusions}
\label{Conclusions}
In this paper we have extended the theory of identification 
via quantum channels to include realistic considerations of 
robustness and security. The former we modelled by channel
uncertainty in both compound and arbitrarily varying cq-channels,
the latter by considering wiretap channels. We considered
these additions of robustness and secrecy constraints separately,
and eventually both of them together. These notions generalize 
the ones presented in \cite{BD17} for classical channels, and we found 
capacity characterizations quite analogous to those of \cite{BD17}. 
There we have also given applications for using ID-codes in 
the secure and robust setting; these applications evidently 
extend to cq-channels.

The first, visible difference in the results resides in the 
fact that while the classical theory in all 
variants essentially yields single-letter formulas, the analogues for 
quantum channels are to a large part multi-letter formulas 
that elude efficient computation, as already seen in the case of
cq-channels considered here, let alone for general quantum
channels.

Secondly, as has been stressed from the beginning of the
theory of identification via quantum channels, it comes naturally
in two flavours, simultaneous \cite{Loeber:PhD} and 
non-simultaneous \cite{AhlswedeWinter:ID-q} identification. 
This is because in quantum mechanics the different 
tests for the various messages correspond to measurements that
are not necessarily compatible, which is an entirely non-classical 
phenomenon. Since it has an additional constraint
on the decoder, the simultaneous ID-capacity is always upper
bounded by the non-simultaneous ID-capacity; thus, it is
desirable to prove the direct coding theorems for the former,
and the converses for the latter. We do this here, and find 
that simultaneous and non-simultaneous ID-capacities coincide
in all the models considered, generalizing the result of \cite{AhlswedeWinter:ID-q}
for ideally known and secrecy-free cq-channels. It should be 
noted, however, that for general quantum channels, a gap between
simultaneous and non-simultaneous ID-capacities is expected,
cf. \cite{winter:survey}.

The converses for the non-simultaneous ID-capacities are 
considerably more difficult than their classical and simultaneous
analogues. As a matter of fact, those can be obtained by general 
information spectrum and resolvability methods, while the converses 
and dichotomy theorems of the non-simultaneous cq-versions require 
genuine quantum generalizations of resolvability ideas, as is
already evident in the matrix concentration bounds from \cite{AhlswedeWinter:ID-q}. 
In the present paper, an interesting case is that of the compound 
channel (Theorem \ref{cqid}), where the converse proof is specifically 
adapted to the channel model, and it follows a completely different
idea from the one known for classical channels.
Another manifestation of the different character of classical
and quantum information is the form of the maximum mutual 
information of a cq-channel with two output states $\mu$-close 
in trace norm (Lemma \ref{Lemma 3}). We need this technical bound 
to argue that ID-secrecy implies wiretap communication secrecy.
While in Remark 4.8 it is shown that for classical channels this
maximum information is precisely $\mu$ (cf. also \cite{AZ95}), the 
analysis for cq-channels is not only much more involved, it also
only yields a very different-looking upper bound. We wish to 
highlight it as an interesting open problem to determine precisely
what the optimal upper bound in Lemma \ref{Lemma 3} is.

\section*{Acknowledgments}

Holger Boche and Christian Deppe were supported by the Bundesministerium 
f\"ur Bildung und Forschung (BMBF) through Grants 16KIS0118K and 16KIS0117K.
Holger Boche is also partly supported by the Deutsche Forschungsgemeinschaft
(DFG, German Research Foundation) under Germany's Excellence Stra\-tegy EXC-2111 390814868. 
Andreas Winter was supported by the ERC Advanced Grant IRQUAT,
the Spanish MINECO, projects FIS2013-40627-P and FIS2016-86681-P, with
the support of FEDER funds, and the Generalitat de Catalunya, CIRIT project 
2014-SGR-966 and 2017-SGR-1127.

Finally, Holger Boche thanks Freeman Dyson for comments on related issues on 
particle physics and computation, and the IAS Princeton for its hospitality.

\vspace{1cm}

\begin{IEEEbiographynophoto}{Holger Boche}
	received the Dr.rer.nat. degree in pure mathematics from the
	Technische Universit\"at Berlin, Berlin, Germany,
	in 1998, the Dipl.-Ing. and Dr.-Ing. degrees in electrical 
	engineering from the Technische Universit\"at
	Dresden, Dresden, Germany, in 1990 and 1994,
	respectively, and the degree in mathematics from
	the Technische Universit\"at Dresden, in 1992.
	From 1994 to 1997, he was involved in postgraduate studies in mathematics with the 
	Friedrich-Schiller Universit\"at Jena, Jena, Germany. In 1997,
	he joined the Heinrich-Hertz-Institut (HHI) f\"ur Nachrichtentechnik Berlin,
	Berlin. In 2002, he was a Full Professor of Mobile Communication Networks
	with the Institute for Communications Systems, Technische Universit\"at Berlin.
	In 2003, he became the Director of the Fraunhofer German-Sino Laboratory
	for Mobile Communications, Berlin, and the Director of HHI in 2004.
	He was a Visiting Professor with ETH Zurich, Zurich, Switzerland, in
	Winter 2004 and 2006, and KTH Stockholm, Stockholm, Sweden, in Summer
	2005. Since 2010, he has been with the Institute of Theoretical Information
	Technology and a Full Professor with the Technische Universit\"at M\"unchen,
	Munich, Germany. Since 2014, he has been a member and an Honorary Fellow
	of the TUM Institute for Advanced Study, Munich. He is a member of the
	IEEE Signal Processing Society SPCOM and the SPTM Technical Committee.
	He received the Research Award Technische Kommunikation from the Alcatel
	SEL Foundation in 2003, the Innovation Award from the Vodafone Foundation
	in 2006, and the Gottfried Wilhelm Leibniz Prize from the German Research
	Foundation in 2008. He was a corecipient of the 2006 IEEE Signal Processing
	Society Best Paper Award and a recipient of the 2007 IEEE Signal Processing
	Society Best Paper Award. He was elected as a member of the German
	Academy of Sciences (Leopoldina) in 2008 and the Berlin Brandenburg
	Academy of Sciences and Humanities in 2009.
\end{IEEEbiographynophoto}

\begin{IEEEbiographynophoto}{Christian Deppe}
	received the Dipl.-Math. degree in mathematics from the 
	Universit\"at Bielefeld, Bielefeld, Germany, in 1996, and the 
	Dr.-Math. degree in mathematics from the Universit\"at Bielefeld, 
	Bielefeld, Germany, in 1998. 
	He was a Research and Teaching Assistant with the
	Fakult\"at f\"ur Mathematik, Universit\"at  Bielefeld from 1998 to 2010. 
	From 2011 to 2013 he was project leader of the 
	project ``Sicherheit und Robustheit des Quanten-Repeaters''
	of the Federal Ministry of Education and Research at 
	Fakult\"at f\"ur Mathematik, Universit\"at  Bielefeld. 
	In 2014 he was supported by a DFG project at the Institute of 
	Theoretical Information Technology, Technische Universit\"at M\"unchen. 
	In 2015 he had a temporary professorship at the Fakult\"at f\"ur 
	Mathematik und Informatik, Friedrich-Schiller Universit\"at Jena. 
	He is currently  project leader of the project ``Abh\"orsichere Kommunikation 
	\"uber Quanten-Repeater'' of the Federal Ministry of Education and 
	Research at Fakult\"at f\"ur Mathematik, Universit\"at  Bielefeld. 
	Since 2018 he is at the Department of Communications Engineering at 
	the Technical University of Munich.
\end{IEEEbiographynophoto}

\begin{IEEEbiographynophoto}{Andreas Winter}
  received a Diploma degree in Mathematics from the Freie Universit\"at Berlin, 
  Berlin, Germany, in 1997, and a Ph.D. degree (Dr. math.) from the Fakult\"at f\"ur 
  Mathematik, Universit\"at Bielefeld, Bielefeld, Germany, in 1999. 
  He was Research Associate at the University of Bielefeld until 2001, 
  and then with the Department of Computer Science at the University of 
  Bristol, Bristol, UK. In 2003, still with the University of Bristol, 
  he was appointed Lecturer in Mathematics, and in 2006 Professor of 
  Physics of Information. Since 2012 he has been ICREA Research Professor 
  with the Universitat Aut\`onoma de Barcelona, Barcelona, Spain. 
  His research interests include quantum and classical Shannon theory, 
  and discrete mathematics. He is recipient, along with Bennett, Devetak, 
  Harrow and Shor, of the 2017 Information Theory Society Paper Award.
\end{IEEEbiographynophoto}

\end{document}